\documentclass[aps,prx,twocolumn,superscriptaddress,notitlepage,longbibliography,floatfix,nofootinbib]{revtex4-2}

\usepackage[utf8]{inputenc}
\usepackage[T1]{fontenc}
\usepackage{microtype}
\usepackage[dvipsnames]{xcolor}

\usepackage{graphicx}
\graphicspath{{figs/}}          
\usepackage{subcaption}         

\usepackage{tikz}
\usetikzlibrary{arrows.meta}

\usepackage{amsmath,amssymb,bm}
\usepackage{amsthm}
\usepackage{siunitx}
\usepackage{physics}
\usepackage{braket}             

\usepackage{newtxtext,newtxmath} 
\DeclareMathOperator{\polylog}{polylog}

\usepackage{ragged2e}           

\usepackage{algpseudocode}      

\usepackage[colorlinks=true,linkcolor=blue,citecolor=blue,urlcolor=blue]{hyperref}
\usepackage[capitalise]{cleveref}

\newtheorem{lemma}{Lemma}

\newtheorem{theorem}{Theorem}
\newtheorem{corollary}{Corollary}

\theoremstyle{definition}

\newtheorem*{definition*}{Definition}

\newcommand{\EnsembClass}{\rho^{\mathrm{cl}}_{p,q}}

\newcounter{algo}

\def\tcm{T.C.M. Group, Cavendish Laboratory, University of Cambridge, J.J. Thomson Avenue, Cambridge, CB3 0US, UK}
\def\quantinuum{Quantinuum, Terrington House, 13-15 Hills Rd, Cambridge, CB2 1NL, UK}

\begin{document}

\title{Critical non-equilibrium phases from noisy topological memories}

\author{Amir-Reza~Negari}
\email{anegari@pitp.ca}
\affiliation{Perimeter Institute for Theoretical Physics, Waterloo, Ontario N2L 2Y5, Canada}
\affiliation{Department of Physics and Astronomy, University of Waterloo, Ontario N2L 3G1, Canada}

\author{Subhayan~Sahu}
\affiliation{Perimeter Institute for Theoretical Physics, Waterloo, Ontario N2L 2Y5, Canada}

\author{Jan~Behrends}
\altaffiliation[Current address: ]{\quantinuum}
\affiliation{\tcm}

\author{Benjamin~B\'eri}
\affiliation{\tcm}
\affiliation{DAMTP, University of Cambridge, Wilberforce Road, Cambridge CB3 0WA, United Kingdom}

\author{Timothy~H.~Hsieh}
\affiliation{Perimeter Institute for Theoretical Physics, Waterloo, Ontario N2L 2Y5, Canada}
\affiliation{Department of Physics and Astronomy, University of Waterloo, Ontario N2L 3G1, Canada}

\date{\today}

\begin{abstract}
We demonstrate the existence of an extended non-equilibrium critical phase, characterized by sub-exponential decay of conditional mutual information (CMI), in the surface code subject to heralded random Pauli measurement channels.  By mapping the resulting mixed state to the ensemble of completely packed loops on a square lattice, we relate the extended phase to the Goldstone phase of the loop model.  In particular, CMI is controlled by the characteristic length scale of loops, and we use analytic results of the latter to establish polylogarithmic decay of CMI in the critical phase.  We find that the critical phase retains partial logical information that can be recovered by a global decoder, but not by any quasi-local decoder.  To demonstrate this, we introduce a diagnostic called punctured coherent information which provides a necessary condition for quasi-local decoding.

\end{abstract}

\maketitle

\section{Introduction}
\label{sec:intro}

Error correction involves the encoding of information into long-range correlations of many-body states~\cite{Shor1995}, resulting in non-trivial phases of matter such as topological order~\cite{Wen1990, Kitaev_2003}.  Topological codes such as toric code exhibit noise thresholds above which encoded information is no longer decodable~\cite{Dennis_2002}. Interpreting this decodability transition as a phase transition has led to fundamental developments in generalizing the notion of a phase to mixed-states (arising from decoherence)~\cite{bao2023mixedstatetopologicalordererrorfield,Fan_2024,WangLi2024AnomalyOpenQuantumSystems,coser2019classification, SangZouHsieh2024MixedStatePhases,PhysRevX.14.041031,
  ChenGrover2024SeparabilityTransitions,
  ChenGrover2024UnconventionalMixedStateTransition,
  Sang_2025,negari2025spacetimemarkovlengthdiagnostic,
  SohalPrem2025NoisyIntrinsicMixedOrder,
  WangWuWang2025IntrinsicMixedStateTopo,
  EllisonCheng2025TowardClassificationMixedState,
  SangEtAl2025MixedStatePhasesLocalReversibility,
  ZhangXuZhangXuBiLuo2025, ma2023average,Lee2025symmetryprotected, LeeJianXu2023QuantumCriticality,lessa_strong--weak_2024, sala2024spontaneousstrongsymmetrybreaking,wang2025decoherenceinducedselfdualcriticalitytopological,liu2025coherenterrorinducedphase}.  Unlike equilibrium critical points in which the correlation length of local observables diverges, such mixed-state phase transitions induced by noise require fundamentally distinct probes.

In particular, conditional mutual information (CMI) provides a diagnostic of such transitions and governs the local recoverability of a state against local noise. For a tripartition $A{:}B{:}C$ of a state $\rho$, the CMI is
\begin{equation}
    I(A:C|B)_\rho = S(AB) + S(BC) - S(B) - S(ABC),
\end{equation}
where $S(\cdot)$ denotes the von Neumann entropy. Consider an annular tripartition such that $A$ and $C$ are completely separated by $B$ of characteristic length $r$ (Fig.~\ref{fig:summary}).  If $I(A:C|B)_{\rho} \sim \exp \left(- r/ \xi_{\mathrm M}\right)$, then the state is said to have finite \emph{Markov length} $\xi_{\mathrm M}$. Gapped ground states ~\cite{Hastings_2006} and finite-temperature Gibbs states of local Hamiltonians ~\cite{chen2025quantumgibbsstateslocally, kuwahara2024clusteringconditionalmutualinformation, Kato_2019} all have finite Markov length.  At thermal (equilibrium) phase transitions, conventional correlation lengths diverge, but the Markov length is finite. Quantum phase transitions in ground states, on the other hand, are characterized by the divergence of both the correlation and Markov lengths, which coincide for pure states.
In contrast, at mixed-state phase transitions of topological memories subject to local noise~\cite{Dennis_2002, Fan_2024, bao2023mixedstatetopologicalordererrorfield}, the correlation length remains finite whereas the Markov length diverges~\cite{Sang_2025}, a fundamentally non-equilibrium phenomenon. In that setting, finite Markov length controls the quasi-local recoverability of encoded information: when the noise is generated by a short time evolution with a local Lindbladian, as long as the Markov length is finite, the effect of local noise can be undone by a quasi-local recovery channel~\cite{Sang_2025}.

\begin{figure*}
  \centering
  \includegraphics[width=0.9\linewidth]{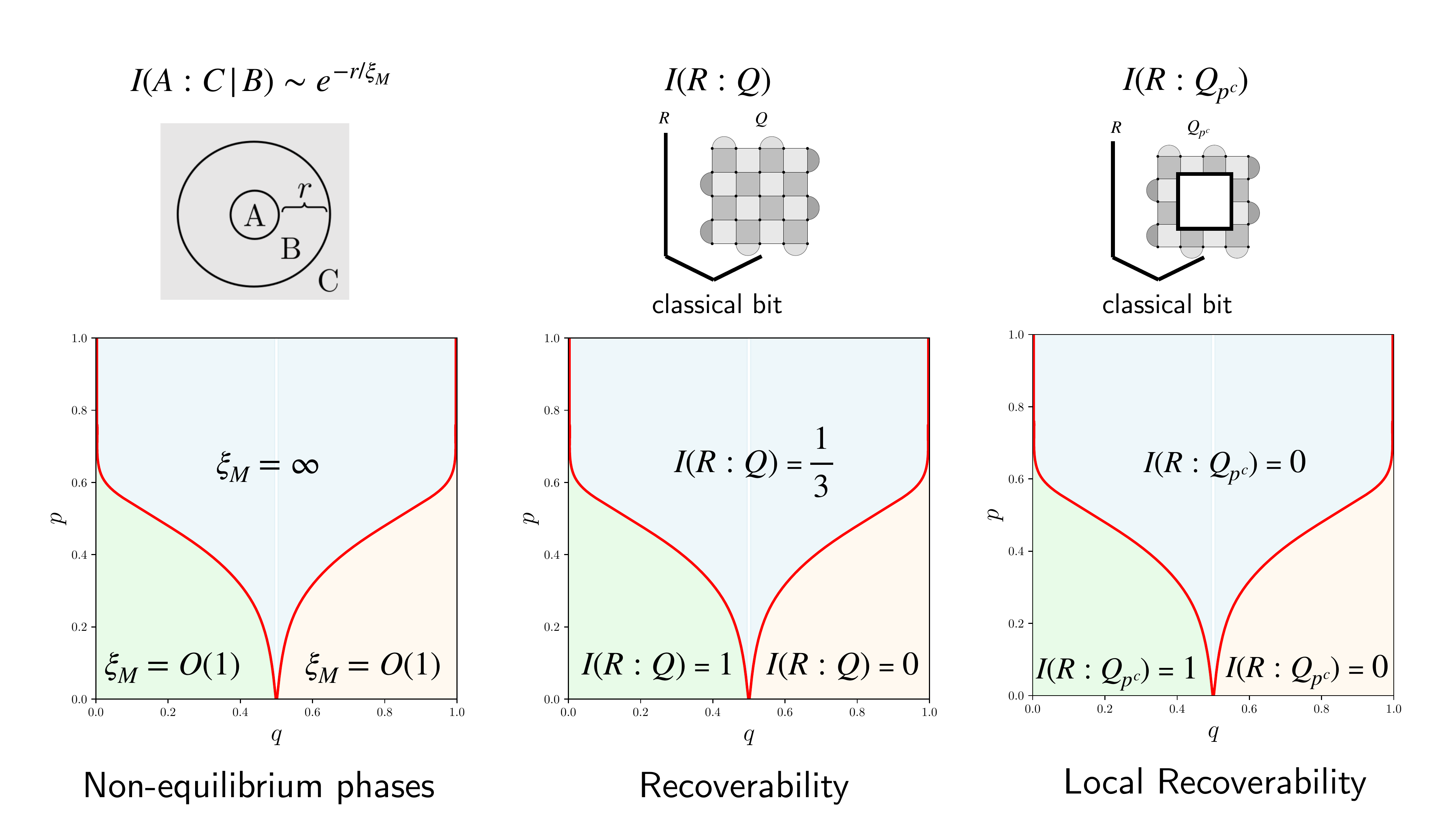}
  \caption{\justifying
  The surface code subject to heralded complete Pauli dephasing has a phase diagram described by the completely packed loop model on a square lattice, and the parameters $0 \leq p,q \leq 1$ control the relative probabilities of $X, Y, Z$ Paulis in the dephasing channel. There are two topologically distinct short-loop phases which exhibit finite Markov length, and an extended Goldstone phase with system-spanning loops, which exhibit diverging Markov length. We also compute the persistence of a classical logical memory using the mutual information~\footnote{This is equivalent to the coherent information $I_{c}(R\rangle Q)\coloneq S(Q)-S(RQ)$ plus the entropy of the reference $S(R)$, which is a constant (equal to $1$ in our setting).} $I(R{:}Q)$ between a classical reference bit $R$ and the dephased code $Q$. Furthermore, we propose a punctured version of the mutual information $I(R{:}Q_{p^c})$ between the reference $R$ and part of the code $Q$ as a probe for local recoverability of the encoded information. Studying the interplay of phases, memory, and recoverability of encoded information in each of these phases, we find that the Goldstone phase is characterized by diverging Markov length, and partial classical memory that is globally but not locally recoverable. In contrast, in the short loop phase (shaded green), $X$-type logical information persists completely, and we show that it can be recovered using a quasi-local channel.
  }
  \label{fig:summary}
\end{figure*}

In equilibrium, there are not only isolated critical points but also extended critical phases such as Luttinger liquids, metallic phases, or the low-temperature Berezinskii–Kosterlitz–Thouless (BKT) phase, which are described by a manifold of renormalization group fixed points~\cite{Sachdev_2011}.  Can there be {\it non-equilibrium} critical phases?  A uniquely non-equilibrium feature for such a critical phase is sub-exponential decay of CMI despite short-ranged correlations in an extended parameter space.  Given such a phase in the context of error correction and topological codes, what are the consequences for decodability?

In this work, we realize a non-equilibrium critical phase in an ensemble obtained from a topological code subject to heralded noise. We show that this ensemble has sub-exponential decay of CMI in a finite parameter range despite short-range correlations, and we relate this regime to the Goldstone phase of a loop model with crossings~\cite{Nahum_2013}.

We then analyze the implications for error correction, finding that in this critical phase, some logical information is globally but not quasi-locally recoverable.  For this purpose, we introduce a new diagnostic called punctured coherent information, which provides a condition for quasi-local decodability, and illustrate its utility in the critical phase.  We also propose an explicit quasi-local decoder in the non-critical phase retaining logical memory.

We note that recent work \cite{wang2025fractionalquantumhallstates} has shown that noisy fractional quantum Hall states can exhibit an extended critical regime of algebraically decaying Renyi correlators, whereas our focus is on establishing a regime of sub-exponential CMI decay and global versus quasi-local decodability. Prior works \cite{LeeJianXu2023QuantumCriticality,Zou_2023,LuZhangVijayHsieh_2023,zhang2025universalpropertiescriticalmixed} have studied critical mixed-states arising from decohering quantum critical pure states, but these have algebraic decay of both CMI and correlations. Finally, we note that while the present work was completed, \cite{vijay2025informationcriticalphasesdecoherence}, which defines ``information critical phases'' also based on diverging Markov length, appeared.

\subsection{Outline and summary of results}

We consider the surface code subjected to a heralded single-qubit Pauli dephasing channel, where a local three-level ``flag" heralds which of the Pauli $X, Y, Z$ bases the physical qubits were dephased in. Equivalently, the channel can be understood as a non-selective local Pauli measurement channel in the surface code. The resulting state is classical (diagonal in a product basis) but can nonetheless retain a classical topological memory.

We map this classical ensemble to a completely packed loop model with crossings (CPLC). In this representation each measurement pattern corresponds to a configuration of loops on the lattice, and each loop imposes a $\mathbb{Z}_2$ constraint on the measurement outcomes. The classical CMI of the ensemble can be expressed directly in terms of the number and structure of these loop constraints.

Our main results are:

\begin{enumerate}
    \item \emph{Critical mixed-state phase and diverging Markov length.}  

    We map CMI of the classical ensemble to properties of the CPLC model. The CMI is given by the number of independent loop-parity constraints that connect $A$ and $C$ when conditioning on $B$. Thus, the CMI in certain tripartitions $A{:}B{:}C$ is controlled by loop spanning numbers and multi-leg (watermelon) correlators in the CPLC.    

    The CPLC has a well-understood phase diagram featuring two short-loop phases and an extended Goldstone critical phase. In the Goldstone phase, the spanning number between opposite boundaries grows logarithmically with system size, while multi-leg correlators have polylogarithmic decay; in the short-loop phases these quantities decay exponentially. We show, analytically and numerically, that the CMI inherits this behavior: for a partition as shown in Fig.~\ref{fig:summary} left panel, the CMI decays exponentially with separation in the short loop phase (implying a finite Markov length), while it has polylogarithmic decay in the Goldstone phase (divergent Markov length).

    \item \emph{Classical memory and (non-)local recoverability.}  

    Finally, we interpret the classical ensemble as a noisy classical memory. After complete $X$ dephasing, the surface code stores a single logical classical bit. We couple this logical bit to a classical reference and track the mutual information between the reference and the noisy code. We introduce a ``punctured'' version of this mutual information---obtained by first tracing out a finite fraction of bulk region of the code---which probes whether decoding can be achieved quasi-locally.

    In the Goldstone phase we find that the input--output mutual information $I(R:Q)$ remains nonzero, showing that the logical bit is retained, but that the punctured mutual information $I(R:Q_{p^c})$ vanishes for any extensive bulk puncture. We establish general constraints relating local recovery maps and punctured coherent information and use them to conclude that in the Goldstone phase no quasi-local decoder can exist. Equivalently, recovery requires a decoder of intrinsically high circuit complexity and cannot be achieved by any quasi-local decoding circuit.

    On the other hand, the phase diagram also contains a logical-preserving short-loop phase in which the encoded bit is not only retained but also quasi-locally recoverable. In this regime we construct an explicit decoder: a recovery channel implementable by a circuit of $\mathrm{polylog}(L)$ range and depth, with failure probability vanishing polynomially as the system size $L\to\infty$. (In the other short-loop phase, the logical information is erased.)

\end{enumerate}

The rest of the paper is organized as follows. In Sec.~\ref{sec:setup} we review the surface code and its parton construction, and define the heralded Pauli dephasing channel and the resulting classical ensemble. In Sec.~\ref{sec:classical_ensemble_cmi} we relate the CMI of this ensemble to loop-counting in the CPLC and establish the existence of a critical mixed-state phase with divergent Markov length. Section~\ref{sec:fault_tolerance} discusses the implications for classical memories, introducing punctured coherent information and constructing an explicit log-local decoder in the short loop regime. We conclude in Sec.~\ref{sec:discussion} with a brief outlook.

\section{Setup}
\label{sec:setup}

\subsection{Surface code and Majorana representation}
\label{sec:surface_code_tensor_network}

We work with the rotated surface code on a square lattice with qubits on the vertices; see Fig.~\ref{fig:surface-code}. Each plaquette $p$ supports either an $X$-type stabilizer $S_p^X = \prod_{i\in p} X_i$ or a $Z$-type stabilizer $S_p^Z = \prod_{i\in p} Z_i$, and the boundary is chosen so that strings of $X$ and $Z$ operators along the edges implement logical $\bar X$ and $\bar Z$. This is the rotated surface-code patch~\cite{Bombin_2007}, closely related to the planar surface codes~\cite{Kitaev_2003}.

\begin{figure}[t]
    \centering
    \includegraphics[width=\linewidth]{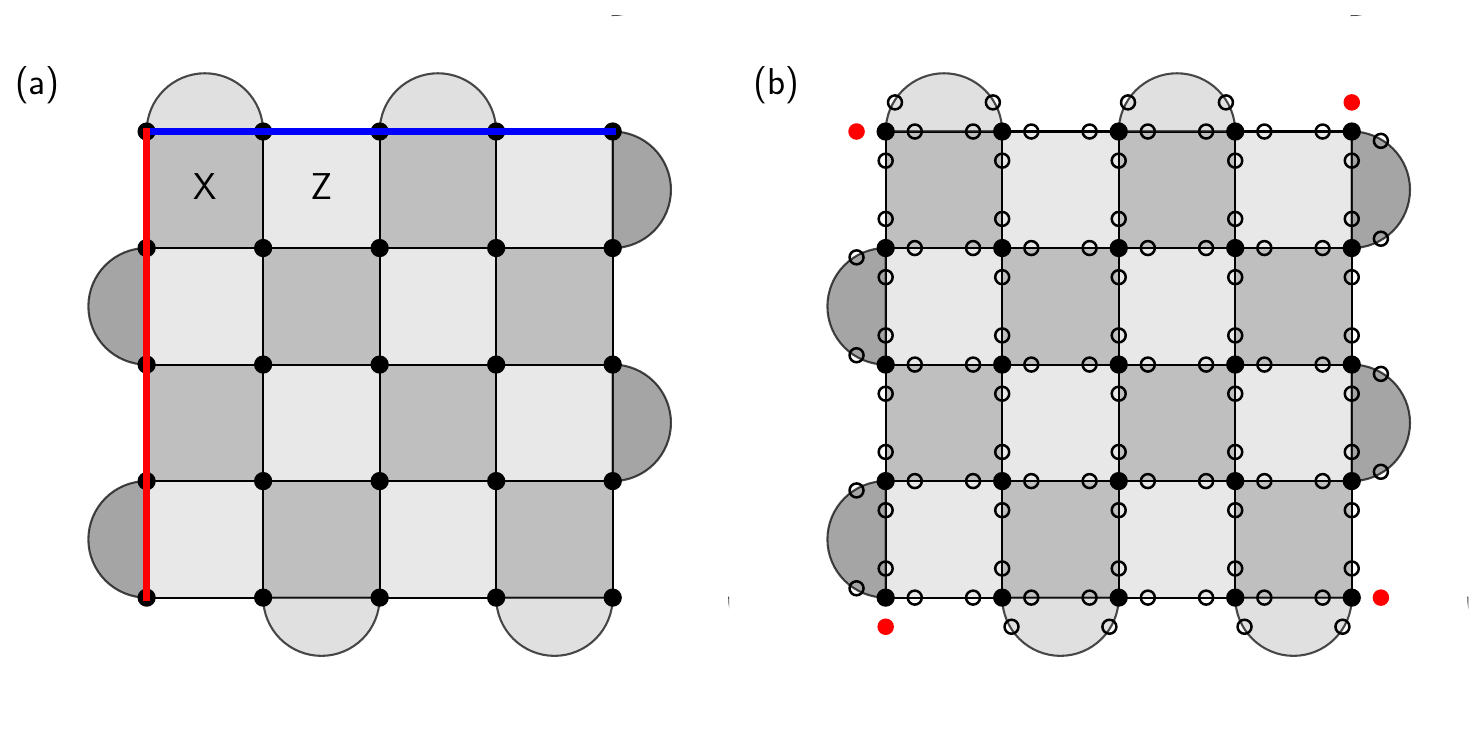}
    \caption{\justifying Left: A $d=5$ rotated surface code with $X$-type (black) and $Z$-type (white) plaquette stabilizers. Logical operators $\bar Z$ (red) and $\bar X$ (blue) run along boundaries. Right: The Majorana--qubit tensor-network representation, with 4 virtual Majorana modes (unfilled circles) attached to each physical qubit leg (filled circles). Neighboring virtual Majoranas are projected into even-parity dimers along edges (see text), while the 4 unpaired corner modes (red circles) encode the logical qubit degrees of freedom.}
    \label{fig:surface-code}
\end{figure}

For our purposes it is convenient to represent the surface-code ground state as a projected fermionic state. This construction is based on the parton representation~\cite{Wen_2003, Kitaev_2006, Bravyi_2018} of a qubit at each lattice site $i$ in terms of four \emph{virtual} Majorana modes $\gamma^i_{\mathrm{up}}$, $\gamma^i_{\mathrm{down}}$, $\gamma^i_{\mathrm{left}}$, and $\gamma^i_{\mathrm{right}}$.

We first define a local 5 legged tensor by specifying a pure state on the Hilbert space of
\emph{(one physical qubit)} $\otimes$ \emph{(four virtual Majoranas)}.
A convenient choice is to take the local tensor to be the projector onto the unique common $+1$ eigenspace of three commuting stabilizers: a purely fermionic parity constraint
\begin{equation}
S^i_1 \;=\; -\,\gamma^i_{\mathrm{up}} \gamma^i_{\mathrm{down}} \gamma^i_{\mathrm{left}} \gamma^i_{\mathrm{right}},
\end{equation}
and two mixed qubit--fermion stabilizers tying physical Pauli operators to Majorana bilinears,
\begin{equation}
S^i_2 \;=\; X^i \,(i\,\gamma^i_{\mathrm{up}}\gamma^i_{\mathrm{right}}),
\qquad
S^i_3 \;=\; Z^i \,(i\,\gamma^i_{\mathrm{up}}\gamma^i_{\mathrm{left}}).
\end{equation}
We choose the above convention for the A sublattice, and on the B sublattice we apply a Hadamard gate $H$ to the physical qubit, which swaps $X^i \leftrightarrow Z^i$ while leaving the Majoranas unchanged.

To construct the global tensor network, we contract the virtual Majorana legs along lattice edges by projecting neighboring modes into \emph{dimers}. If the right leg of site $i$ is joined to the left leg of site $j$, we insert the even-parity (dimer) projector
\begin{equation}
P^{(\eta_{ij})}_{ij} \;=\; \frac{1}{2}\Bigl(1 + \eta_{ij}\, i\,\gamma^i_{\mathrm{right}} \gamma^j_{\mathrm{left}}\Bigr),
\label{eq:dimer-projector}
\end{equation}
and similarly along all other edges. Here $\eta_{ij}=\pm 1$ is a choice of sign (equivalently, an orientation convention) for that bond. After inserting all such bond projectors and contracting all virtual Majorana legs, only the physical qubit legs remain open; the resulting state on the physical legs is the surface-code state defined by this construction.

For a fixed choice of local tensors and a fixed sign pattern $\{\eta_{ij}\}$, the contraction defines a unique physical state (up to overall normalization). However, different sign choices in the dimer projections can generate surface-code stabilizer realizations that differ by signs of plaquette operators (i.e., which plaquettes are constrained to have eigenvalue $+1$ versus $-1$ in the resulting stabilizer state). In particular, to obtain the standard rotated surface code in which all plaquette stabilizers are $+1$, one must choose a compatible sign/orientation convention for the dimer projectors. We follow the convention described in Ref.~\cite{Bravyi_2018}, which fixes the $\{\eta_{ij}\}$ (or equivalently an edge orientation) so that the induced plaquette signs match the usual surface-code stabilizers.

Along the outer boundary, four Majorana legs remain uncontracted. These boundary Majoranas encode the logical degrees of freedom: different ways of pairing them correspond to different logical bases, and in particular realize the logical Pauli operators $\bar X, \bar Y, \bar Z$ as nonlocal bilinears of boundary modes (possibly dressed by boundary dimer variables, depending on convention). The topological information of the code is thus carried by the boundary Majoranas together with the bulk pattern of dimer projections.

This parton (Majorana) construction is particularly useful for tracking the effect of single-qubit Pauli measurements. In the tensor-network picture, a Pauli measurement acts only on the physical leg of a local tensor; fixing that leg (in the appropriate Pauli eigenbasis) is equivalent to imposing a local constraint on the four virtual Majorana legs, thereby selecting one of a small set of allowed fermionic pairings (``rewirings'') inside the tensor. The virtual contractions between neighboring sites are unchanged, so a measurement pattern updates the global state simply by replacing the affected site tensors with new purely fermionic tensors implementing the corresponding local Majorana rewiring. For example, on the sublattice B (which includes the second qubit on the topmost row in Fig.~\ref{fig:surface-code}), an $X$-measurement on site $i$ with outcome $m=\pm1$ (projection $\Pi^{(X)}_{m}=\frac12(\mathbb{I}+m X^i)$) replaces $X^i$ by $m$ in the mixed stabilizer $S^i_2=X^i(i\gamma^i_{\mathrm{up}}\gamma^i_{\mathrm{right}})$, yielding the Majorana constraint $i\gamma^i_{\mathrm{up}}\gamma^i_{\mathrm{right}}=m$; together with the parity $S^i_1=-\gamma^i_{\mathrm{up}}\gamma^i_{\mathrm{down}}\gamma^i_{\mathrm{left}}\gamma^i_{\mathrm{right}}=+1$, this implies $i\gamma^i_{\mathrm{down}}\gamma^i_{\mathrm{left}}=m$. Thus the physical leg is removed and the site reduces to a Gaussian tensor pairing the virtual Majoranas into dimers $(\mathrm{up},\mathrm{right})$ and $(\mathrm{down},\mathrm{left})$. The same logic applies for $Y$ and $Z$ measurements: fixing the physical leg replaces the relevant mixed qubit--Majorana stabilizer(s) by purely Majorana bilinears, selecting a definite local pairing pattern, which are listed in Fig.~\ref{fig:cplc_rules}. In the rest of the paper we exploit this simplification: bulk operations acting on physical qubits (such as the heralded Pauli dephasing channel) reduce to local tensor replacements that implement these rewiring rules, while the pattern of virtual dimer projections stays fixed; contracting the resulting purely fermionic network yields either constrained loop ensembles (in the fully dephased classical description) or effective boundary states defined on the uncontracted Majoranas, providing a direct bridge between the properties of the ensemble and the loop model.

\subsection{Heralded Pauli dephasing channel and the  ensemble}
\label{subsec:classical-ensemble-def}

We now define the heralded noise model and the resulting classical ensemble. Heralded erasure channels, where the location of erasure or loss event is known have been extensively studied~\cite{Grassl1997ErasureCodes, Knill2005LargeDetectedErrors, Stace2009LossThresholds}. In this work, we introduce a heralded dephasing channel, defined as follows. 

On each site we introduce a three-level ``flag'' that indicates the local Pauli operator on which the dephasing is performed. The physical Hilbert space is
\(
  \mathcal{H}_{\mathrm{s}} \cong (\mathbb{C}^2)^{\otimes n},
\)
and the flag space is
\(
  \mathcal{H}_{\mathrm{f}} \cong \bigl(\mathrm{span}\{\ket{0},\ket{-1},\ket{+1}\}\bigr)^{\otimes n}.
\) On each physical qubit of the surface code, we apply a local CPTP map $\mathcal N_{p,q}$, parametrized by two probabilities $p,q \in [0,1]$, 
\begin{align}
    \mathcal N_{p,q} (\rho) = &(1-p)(1-q) \mathcal{D}_X (\rho) \otimes \ket{0}\!\bra{0}_{\mathrm{f}} \nonumber \\
    &+ p \mathcal{D}_Y(\rho) \otimes \ket{+1}\!\bra{+1}_{\mathrm{f}} \nonumber \\
    &+ (1-p)q \mathcal{D}_Z(\rho) \otimes\ket{-1}\!\bra{-1}_{\mathrm{f}},
\end{align}
which outputs on $\mathcal{H}_{\mathrm{s}}\otimes\mathcal{H}_{\mathrm{f}}$ and $\mathcal{D}_{\sigma}(\cdot)$ is the complete dephasing channel in the $\sigma = \{X, Y, Z\}$ basis, i.e., $\mathcal{D}_{\sigma}(\rho) = \frac{1}{2}\rho + \frac{1}{2}\sigma \rho \sigma$. In words, this channel describes a flagged dephasing channel, where the local dephasing is done with probabilities $(1-p)(1-q), ~p,~(1-p)q$ in the $X,Y,Z$ bases respectively.

After the application of this channel, the resulting ensemble is,
\begin{equation}
  \rho_{p,q}
  :=\mathcal{N}_{p,q}^{\otimes n}(\rho_{\mathrm{code}}),
  \label{eq:state-classical}
\end{equation}
where $\rho_{\mathrm{code}}$ is a surface-code ground state. The state $\rho_{p,q}$ is diagonal in a local product basis; it has no entanglement, but, as we show, it can exhibit nontrivial long-range structure quantified by CMI.

The ensemble $\rho_{p,q}$ can also be interpreted as a result of non-selective single qubit Pauli measurements on the surface code. When such measurements are performed on a subset of the qubits, the remaining unmeasured qubits undergo an entanglement phase transition, which was studied in~\cite{Negari_2024}. Robustness of the stored logical information to single qubit measurements on the surface code was studied in~\cite{Botzung2025robustness, eckstein2025learningtransitionstopologicalsurface, lee2024randomlymonitoredquantumcodes}.

\begin{figure}[h]
    \centering

    \begin{subfigure}{\columnwidth}
        \centering
        \includegraphics[width=\linewidth]{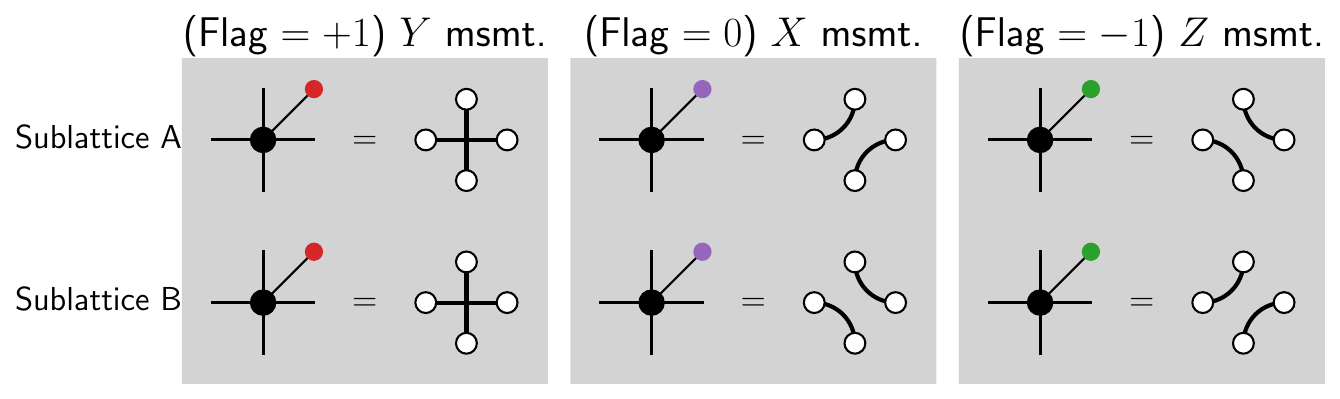}

        \caption{\justifying Local Majorana pairings in the hybrid fermion--qubit tensor for different flag values (measurement bases) on the two sublattices.}
        \label{fig:cplc_rules}
    \end{subfigure}

    \hfill

    \begin{subfigure}{\columnwidth}
        \centering
        \includegraphics[width=1.01\linewidth]{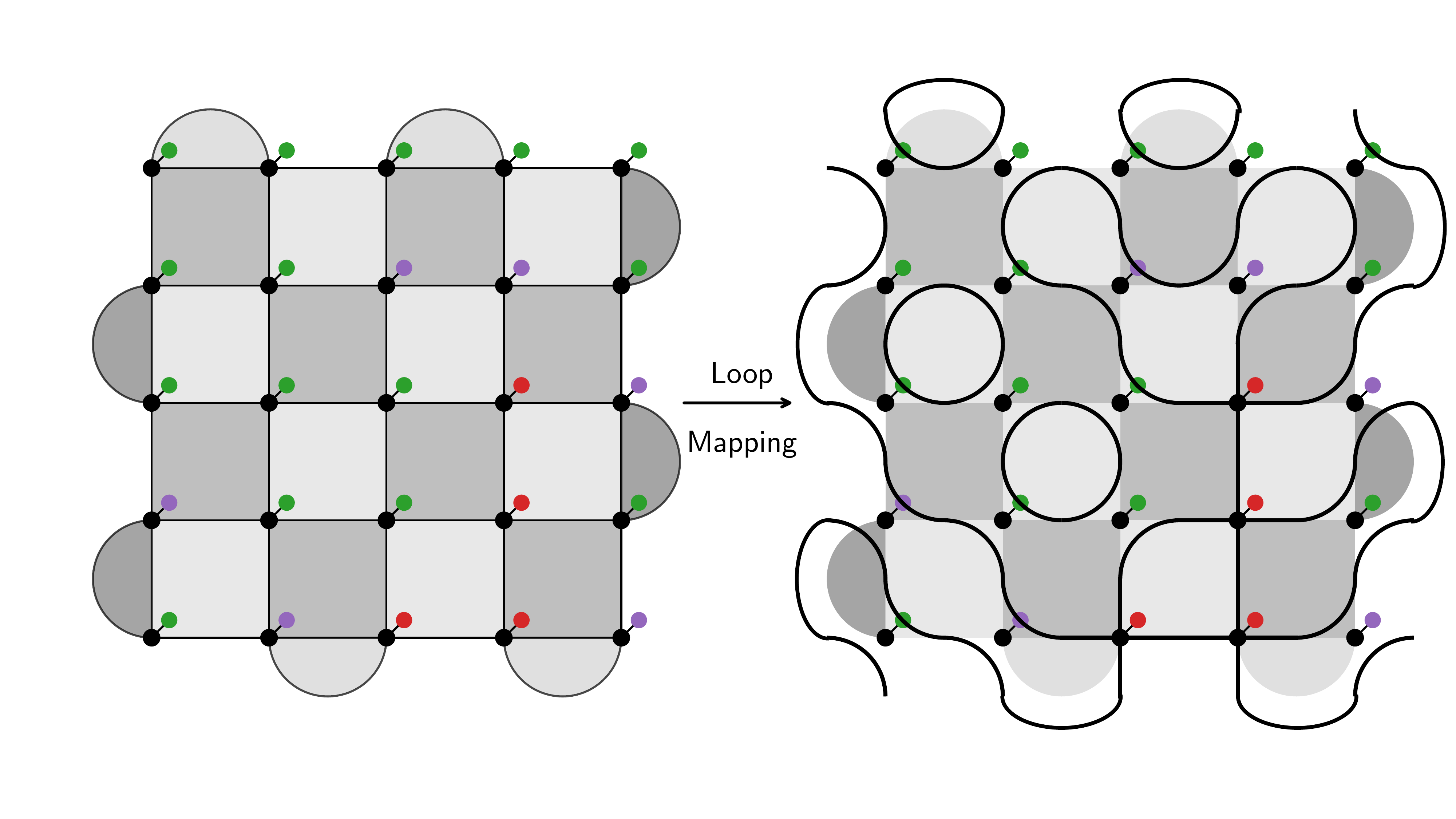}
        \caption{\justifying Gluing these local tensors via Gaussian contractions yields a configuration of completely packed loops on the lattice.}
        \label{fig:cmi_loop_mapping}
    \end{subfigure}

    \hfill

    \begin{subfigure}{\columnwidth}
        \centering
        \includegraphics[width=1.01\linewidth]{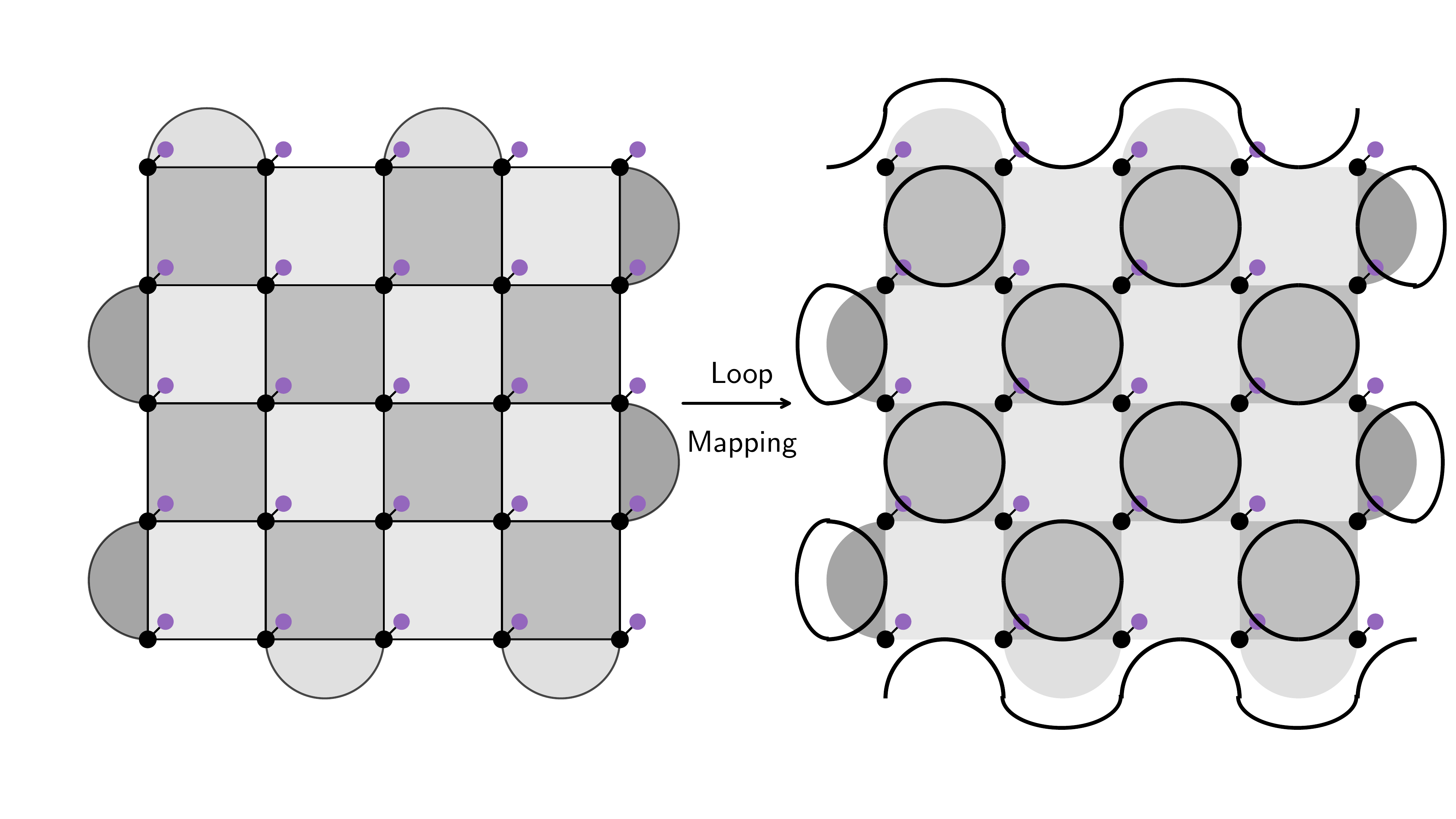}
        \caption{\justifying In the clean memory limit $p = q = 0$, the loop mapping exhibits only plaquette loops, and open strands along the top and bottom boundaries describing the representative logical operators.}
        \label{fig:cmi_x_logical}
    \end{subfigure}
    
    \caption{\justifying Mapping from the dephased surface code to loop model.}
    \label{fig:cplc_mapping}
\end{figure}

\section{Phase diagram of CMI and completely packed loop model}
\label{sec:classical_ensemble_cmi}

In this section we analyze the CMI in the classical ensemble $\rho_{p,q}$, and relate it to the loop-spanning properties of the completely packed loop model with crossings (CPLC). This allows us to identify the Goldstone phase of the loop model with a critical mixed-state phase with divergent Markov length.

\subsection{Mapping to loop model}

We label lattice sites by $i=1,\ldots,n$ and denote by
$f=(f_1,\ldots,f_n)\in\{-1,0,+1\}^n$
a configuration of local flags, which specifies the Pauli basis in which each data qubit is dephased. Since the noise channel acts independently on each site, the probability of a given flag configuration factorizes as $p(f)=\prod_i p_{f_i}$, with $p_{f_i}=(1-p)(1-q),\,p,$ or $(1-p)q$ for $f_i=0,+1,-1$, respectively. $\rho_{p,q}$ is block diagonal and can be written as
\begin{equation}
\rho_{p,q}
= \sum_{f} p(f)\, \ket{f}\!\bra{f}_{\mathrm{f}} \otimes \sigma_f,
\qquad
\sigma_f = \mathcal{D}_f\!\left(\rho_{\mathrm{code}}\right),
\label{eq:cq-form}
\end{equation}
where $\mathcal{D}_f=\bigotimes_i \mathcal{D}_{f_i}$ denotes complete dephasing of the physical qubits in the Pauli bases specified by $f$. For fixed $f$, every physical qubit is dephased (equivalently, nonselectively measured) in a single-qubit Pauli basis
$\sigma_i\in\{X_i,Y_i,Z_i\}$.
Thus $\sigma_f=\mathcal D_f(\rho_{\mathrm{code}})$ is a classical mixture over measurement trajectories or Pauli eigenvalue strings
$m_i\in\{\pm1\}$, and it is maximally mixed on the subspace compatible with the stabilizer constraints that remain after imposing these local measurements.
Equivalently, if we denote by $\mathcal S_f$ the subgroup of the surface-code stabilizer group that is diagonal in the measured product basis specified by $f$, then
\(
  \sigma_f \ \propto\ \prod_{g\in \mathcal S_f}\frac{1+g}{2},
\)
so the classical distribution of outcomes is uniform over assignments satisfying all independent parity checks $\mathcal S_f$.

Now we use the Majorana fermionic representation of surface code to show $\mathcal S_f$ constraint are given by loop configurations in the loop model. Following the discussion in Sec.~\ref{sec:surface_code_tensor_network}, for a given $f$, the physical leg of each local qubit-fermion tensor is projected onto a $\sigma_i$ eigenstate, leaving a purely fermionic (Gaussian) stabilizer state on that site.
Depending on whether $\sigma_i=X,Y,$ or $Z$ (and depending on the sublattice A or B, which are related by a Hadamard), the remaining fermionic state enforces one of a small set of local ``rewirings'' (pairings) among the four virtual Majoranas; see Fig.~\ref{fig:cplc_mapping}(a).

Contracting neighboring sites with the fixed Gaussian bond projectors glues these local pairing patterns into configurations of completely packed loops on the square lattice~\cite{Negari_2024}, as illustrated in Fig.~\ref{fig:cplc_mapping}(b). 
Each measurement trajectory with $m_i=\pm1$ fixes the signs associated with the fermionic pairings at each site, but the overall connectivity pattern remains the same for all trajectories consistent with the flag pattern $f$.

Each closed loop $L$ appearing in a particular configuration yields a $\mathbb Z_2$ constraint on the classical measurement outcomes associated with that loop.
The origin of this constraint can be understood by multiplying the local qubit--fermion stabilizers associated with the sites and the bond projectors along the loop, which cancels all internal Majorana operators (each virtual Majorana appears exactly twice), leaving a purely physical Pauli stabilizer operator that survives the local measurements.
Importantly, this physical operator is supported only on those sites where the loop contributes an \emph{odd} number of physical Pauli factors: if the loop passes through a site twice, the corresponding Pauli appears twice and cancels ($\sigma_i^2=\mathbb I$), so that site does not appear in the support.

Equivalently, the loop defines a stabilizer element in the subgroup $\mathcal S_f$ that is diagonal in the measured product basis,
\(
  \prod_{i\in \mathrm{supp}(L)}\!\sigma_i \in \mathcal S_f,
\)
where $\mathrm{supp}(L)$ is the set of sites where the induced physical Pauli appears once and not twice. Evaluating this operator on an outcome string gives the corresponding parity constraint that the measurement outcomes must satisfy. Note that the product of Pauli outcomes around a loop can acquire a sign. In particular,
\begin{equation}
  \prod_{i\in \mathrm{supp}(L)} m_{\sigma_i} = s(L)\in\{\pm1\},
  \label{eq:loop-parity}
\end{equation}
where $s(L)$ is the stabilizer phase of the loop operator. A convenient way to express it in the present setting is via the number of $Y$-type factors that survive on the loop: let $N_Y(L)$ be the number of vertices $i\in \mathrm{supp}(L)$ with $\sigma_v=Y$ (so vertices visited twice by the loop do not contribute). Then the overall phase is
\(
  s(L)=(-1)^{N_Y(L)/2},
\)
since each surviving $Y$ contributes a factor of $i$ when rewriting the loop operator as a product of underlying $X$- and $Z$-type stabilizers, and for closed loops $N_Y(L)$ is even. Note, the logical operators are represented as open-strands in the loop configuration, which also imposes a parity constraint on the Pauli outcomes along the strand (see Fig.~\ref{fig:cmi_x_logical}).

The probability of each measurement trajectory $\{m_i\}$ for a particular flag configuration vanishes unless all loop parities \eqref{eq:loop-parity} are satisfied.
Moreover, conditioned on a fixed loop configuration, these loop parities generate all surviving constraints, so the allowed outcomes are uniformly distributed over the satisfying constraints.
These statements were already derived within the stabilizer formalism; we emphasize the loop interpretation because it makes the connection to the CPLC model explicit.

\begin{figure}[t]
    \centering
    \includegraphics[width = 0.7\columnwidth]{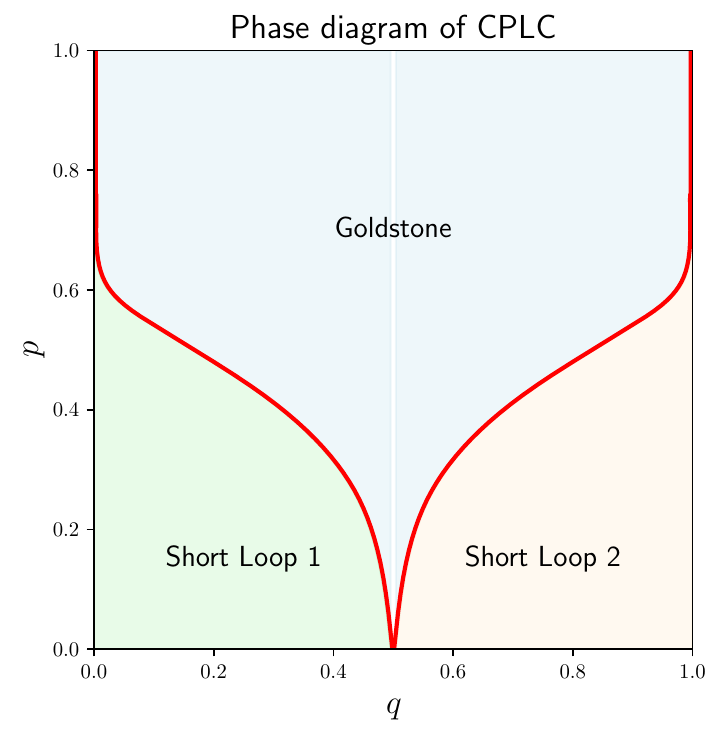}

    \caption{\justifying Phase diagram for the CPLC model~\cite{Nahum_2013}. 
    }
    \label{fig:cplc_phase_diagram}
\end{figure}

The CPLC model has a well established phase diagram~\cite{Nahum_2013}, which is shown in Fig.~\ref{fig:cplc_phase_diagram}. In particular, as a function of $p, q$, the model hosts two topologically distinct ``short loop'' phases where the typical loops are short, and a ``Goldstone'' phase where the typical loops span across the system. The phase diagram can be described field theoretically as a $O(n)$ magnet coupled to a $\mathbb{Z}_2$ gauge field in the limit $n\to 1$. The Goldstone phase arises as a Kosterlitz-Thouless type superfluid whose quasi-long range order is destroyed by a proliferation of $\mathbb{Z}_2$ vortices in the short loop phases~\cite{Nahum_2013}. Next, we relate the loop spanning properties of the phases to their CMI behavior.

\subsection{CMI computation}

We now fix a tripartition $A{:}B{:}C$ of the physical qubits. Because the flags are classical variables that are statistically independent between sites, and hence between the three regions, the CMI of $\rho_{p,q}$ simplifies considerably. The block-diagonal structure in the flag basis implies that all von Neumann entropies entering the definition of $I(A{:}C\,|\,B)$ decompose into a flag average of entropies of $\sigma_f$, following Eq.~\ref{eq:cq-form}. The classical flag contributions cancel exactly in the CMI, leaving
\begin{equation}
I(A{:}C\,|\,B)_{\rho_{p,q}}
= \sum_f p(f)\, I(A{:}C\,|\,B)_{\sigma_f}.
\label{eq:CMI-avg}
\end{equation}
Thus, the CMI of the full classical ensemble is simply the average, over flag configurations, of the CMI computed for each fixed dephasing pattern $f$.

Equation~\eqref{eq:CMI-avg} reduces the analysis of conditional correlations in $\rho_{p,q}$ to the study of the classical distributions $\sigma_f$. In the following, we evaluate $I(A{:}C\,|\,B)_{\sigma_f}$ for a fixed flag configuration using the Majorana representation, which maps each $\sigma_f$ to a constrained loop model determined by the spatial pattern of local Pauli dephasing, as described in the earlier subsection.

For a particular non-overlapping tripartition $A{:}B{:}C$ of the qubits on the surface code patch, the flag configurations can be organized into two cases, one where there are no loops (constraints) that simultaneously touch $A$ and $C$, and the other where there are loops that do, as can be seen in the following schematic in Fig.~\ref{fig:cmi_loops}.

\begin{figure}[h]
    \centering

    \begin{subfigure}{\columnwidth}
        \centering
        \includegraphics[width=0.8\linewidth]
        {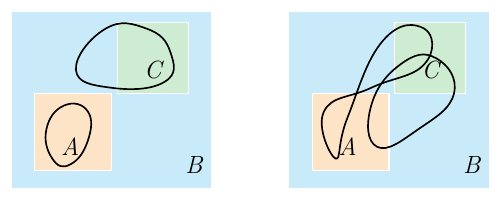}
        
        \caption{\justifying Loop configurations that do not (left) and do (right) contribute to the CMI $I(A:C|B)$.}
        \label{fig:cmi_loops}
    \end{subfigure}

    \hfill

    \begin{subfigure}{\columnwidth}
        \centering
        \includegraphics[width=1.01\linewidth]{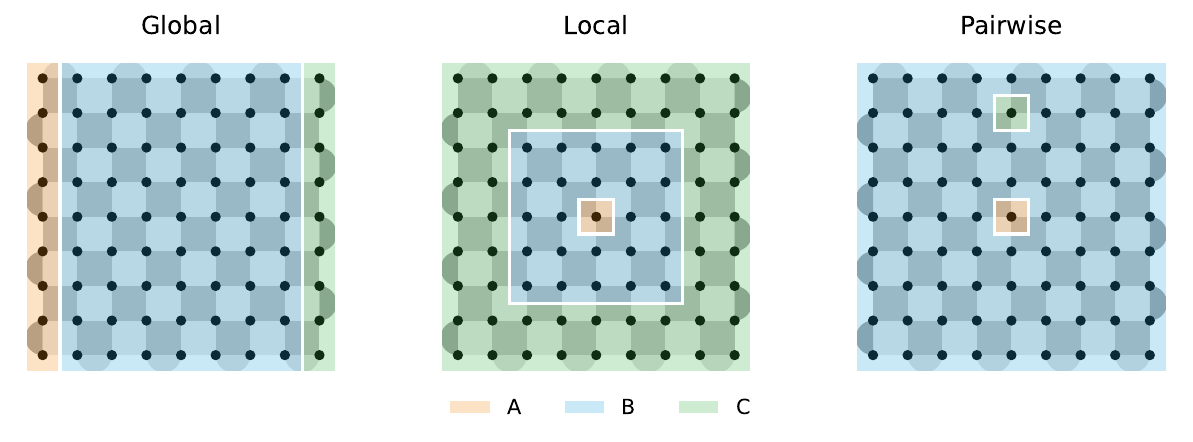}
        \caption{\justifying Surface code geometries of tripartition for which the CMI is computed for the classical ensemble: ``global", ``local", and ``pairwise".}
        \label{fig:cmi_geometries}
    \end{subfigure}

    \hfill

    \begin{subfigure}{\columnwidth}
        \centering
        \includegraphics[width=0.85\linewidth]{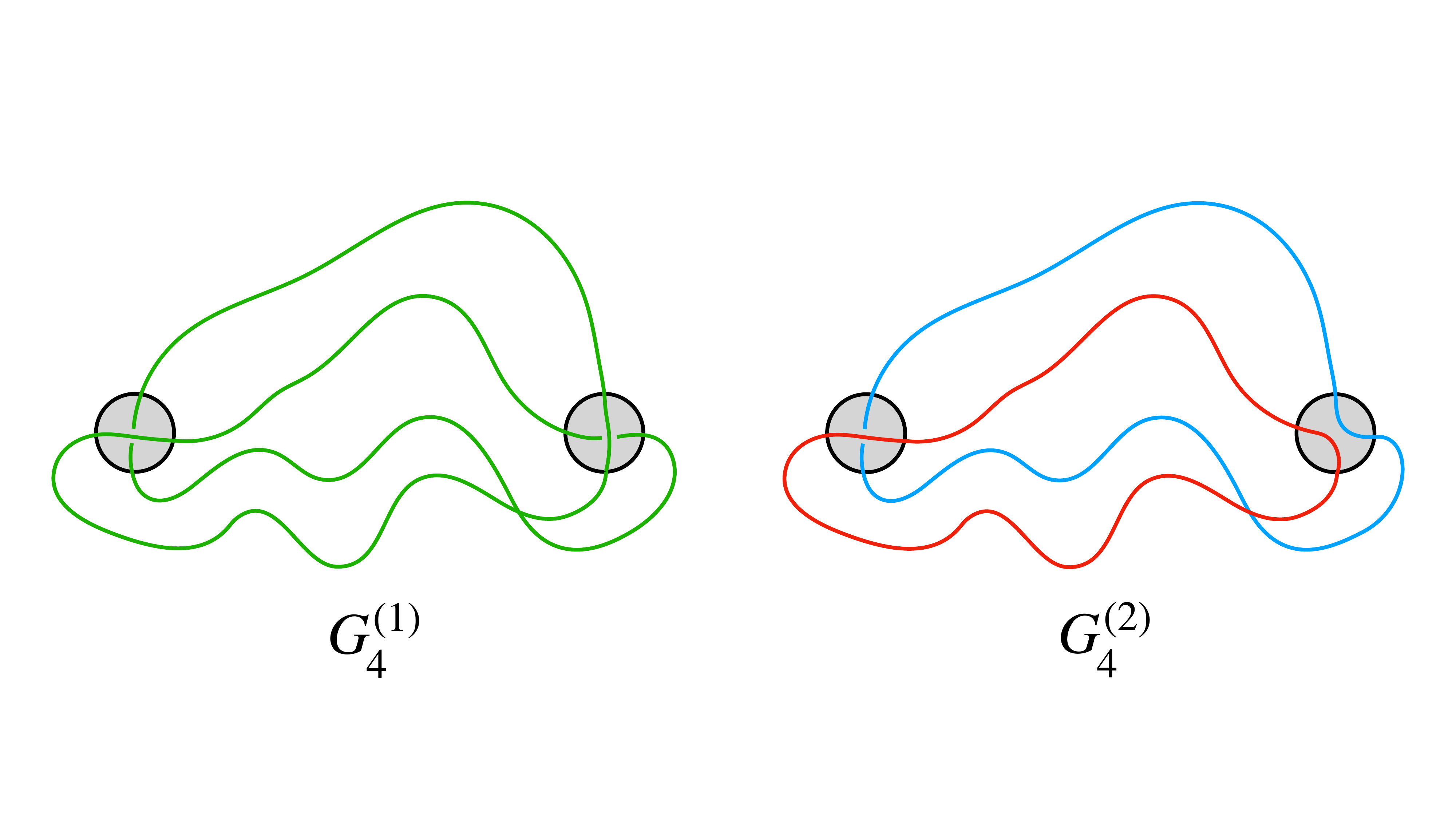}
        \caption{\justifying The two topologically distinct diagrams for the ``pairwise" strand counting. Only $G_4^{(2)}$ contributes to the CMI.}
        \label{fig:cmi_pairwise_diagrams}
    \end{subfigure}
    
    \caption{\justifying Loop configurations and geometries for the CMI computation.}
    \label{fig:cmi_geometries_contributions}
\end{figure}

First note that, only those flag patterns whose induced loop configuration contains at least one loop that touches both $A$ and $C$ can contribute to the CMI.
If every loop that intersects $A$ is entirely contained in $A\cup B$, and every loop that intersects $C$ is entirely contained in $B\cup C$, then $I(A:C\mid B)=0$ for that configuration.
In contrast, a loop that traverses from $A$ to $C$ through $B$ generates a $\mathbb Z_2$ parity constraint relating measurement outcomes in $A$ and $C$ that is not removed by conditioning on $B$.
The CMI (in log base 2) for a fixed flag pattern is therefore equal to the number of \emph{independent} $A$--$C$ parity constraints that survive after conditioning on $B$. Consequently, the CMI is exactly governed by the loop statistics of the CPLC model.

For a fixed flag pattern $f$, write all $\mathbb Z_2$ parity constraints of $\sigma_f$ as a binary matrix
$C_{ABC}(f)$ (rows = constraints, columns = qubits).  Conditioning on $B$ amounts to fixing the $B$ bits, so for the purpose of counting independent constraints we may delete
the $B$ columns.  Let $C_{AC}(f)$ be the resulting matrix, and let $C_A(f)$ and $C_C(f)$ denote its restrictions
to the $A$ and $C$ columns, respectively.  Then the number of \emph{irreducible} constraints that have support on
both $A$ and $C$ (i.e.\ not generated by $A$-only and $C$-only constraints) is
\begin{equation}
\begin{aligned}
I(A{:}C\mid B)_{\sigma_f}
&=\rank_{\mathbb F_2} C_A(f)
 + \rank_{\mathbb F_2} C_C(f) \\
&\quad - \rank_{\mathbb F_2} C_{AC}(f).
\end{aligned}
\label{eq:CMI_indep_constraints_rank}
\end{equation}
where the ranks are obtained by Gaussian elimination over $\mathbb F_2$. 

We begin with the \emph{pairwise} geometry in Fig.~\ref{fig:cmi_geometries}, where $A$ and $C$ each consist of a single qubit and $B$ is the rest. In this case, the only way to obtain a constraint between $A$ and $C$ after conditioning on $B$ is for the loop configuration to contain \emph{two distinct loops} that connect $A$ to $C$ (see Fig.~\ref{fig:cmi_pairwise_diagrams} right).  There are other possibilities---$A$ and $C$ may not be connected at all, or they may be connected by a single loop that visits each site once or twice (see Fig.~\ref{fig:cmi_pairwise_diagrams} left)---but in either case no parity constraint with net support on both $A$ and $C$ survives after conditioning on $B$, and therefore these configurations do not contribute to $I(A{:}C\mid B)_{\sigma_f}$.

It is useful to translate this into standard loop-model language.  The relevant observable is the \emph{four-leg} (``watermelon'') correlator $G_4(r)$~\cite{Nahum_2013}, defined as the probability that two marked sites separated by distance $r$ are connected by a total of four strands.  Importantly, the event ``four strands connect $A$ and $C$'' has two distinct topological components, depending on how the strands are paired internally at the marked sites.  In one component, the four strands belong to a \emph{single} loop that visits both $A$ and $C$ twice, so that $A$ and $C$ are connected through one large loop.  In the other component, the four strands split into \emph{two distinct} loops, each of which visit $A$ and $C$ once.  We therefore decompose
\(
G_4(r)=G_4^{(1)}(r)+G_4^{(2)}(r),
\)
where $G_4^{(1)}$ is the single-loop contribution and $G_4^{(2)}$ is the two-loop contribution, which are shown in Fig.~\ref{fig:cmi_pairwise_diagrams}.

In the pairwise geometry, only the two-loop component produces an $A$--$C$ parity constraint that survives conditioning on $B$, and hence
\begin{equation}
I(A{:}C\mid B)_{\mathrm{pairwise}}
=  G_4^{(2)}(r),
\label{eq:CMI_equals_G4two}
\end{equation}
i.e., the pairwise CMI is exactly the flag-averaged probability that $A$ and $C$ are connected by two distinct loops.  As argued in Ref.~\cite{Nahum_2013}, the two-loop component $G_4^{(2)}(r)$ has the same long-distance scaling as the full four-leg correlator $G_4(r)$, so the pairwise CMI inherits the same asymptotic scaling behavior. Hence \(
I(A:C\,|\,B)\asymp \,G_4(r),
\)
with scaling
\begin{equation}
I(A:C\,|\,B)_{\text{pairwise}}\sim
\begin{cases}
r^{-2x_4}, & \text{transition},\\[3pt]
e^{-r/\xi_{\mathrm{loop}}}, & \text{short loop},
\\[3pt]
[\log(r/l)]^{-12}, & \text{Goldstone},
\end{cases}
\label{eq:CMI_equals_G4two_scaling}
\end{equation}
where $x_4$ is reported numerically in Ref.~\cite{Nahum_2013}.

\begin{figure}[h]
    \centering
    \begin{subfigure}{\columnwidth}
        \centering
        \includegraphics[width=1.02\linewidth]{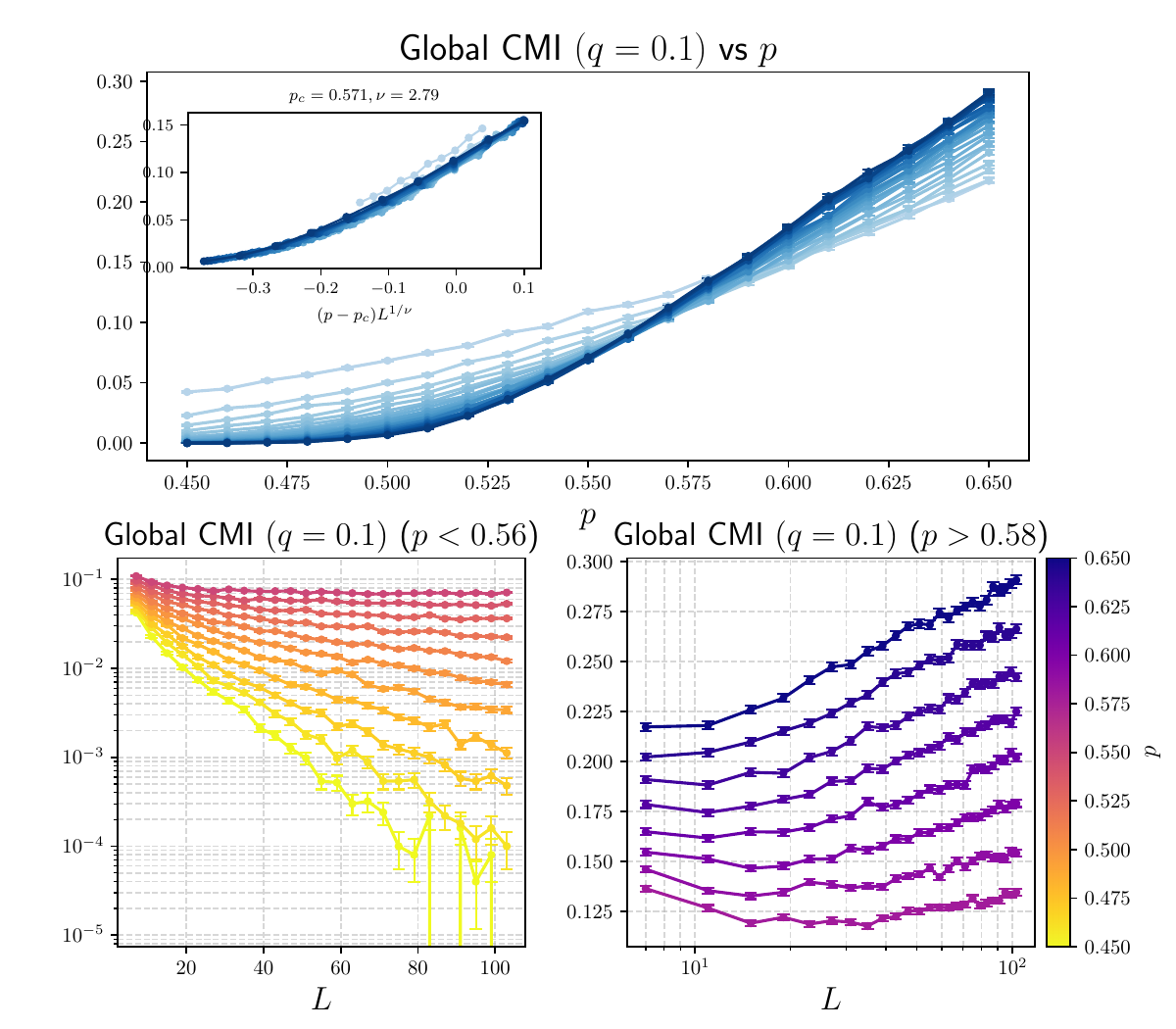}
        \caption{\justifying The global CMI as a function of $p$ for a fixed $ q = 0.1$, and for different sizes of the surface code patch of $L\times L$, with $L = 21-161$. There is a clear transition at around $p_c \approx 0.571$, with the critical exponent $\nu \approx 2.79$. The traces are averaged over $5\times10^4$ samples.}
        \label{fig:global_cmi_transition}
    \end{subfigure}

    \hfill
    
    \begin{subfigure}{\columnwidth}
        \centering
        \includegraphics[width=1\linewidth]{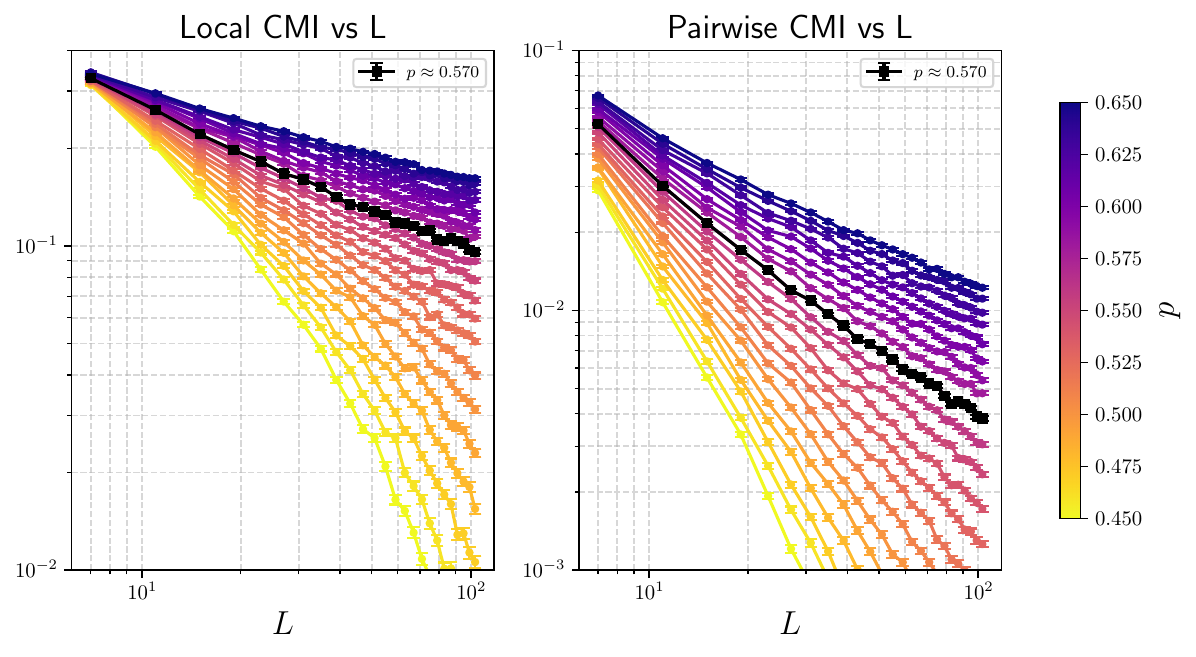}
        \caption{\justifying The CMI of the classical ensemble $\EnsembClass$ for local (left) and pairwise (right) geometries, with $d_{AC} = L/2$ is plotted as a function of $L$ on the log-log scale. Both decay exponentially with $L$ for $p<p_c$, and sub-exponentially with $L$ for $p>p_c$.}
        \label{fig:local_pair_cmi_scaling}
    \end{subfigure}
  
    \caption{\justifying 
    Numerical results for the CMI for different geometries for a fixed $q = 0.1$, with varying $p$ and sizes of the surface code patch $L$, which diagnose the short loop to Goldstone phase transition in the CPLC model.}
    \label{fig:cplc_cmi_transition}
\end{figure}

For the \emph{local} and \emph{global} geometries, the CMI is most conveniently evaluated by directly implementing Eq.~\eqref{eq:CMI_indep_constraints_rank} for each sampled flag pattern: we determine the loop-parity constraint space and compute the number of independent $A$--$C$ constraints that remain after conditioning on $B$ using the stabilizer formalism (Gaussian elimination over $\mathbb F_2$). Figure~\ref{fig:cplc_cmi_transition} shows the resulting behavior at fixed $q=0.1$ as a function of $p$.  For the pairwise geometry, this numerical procedure provides a direct check of the loop-correlator prediction Eq.~\eqref{eq:CMI_equals_G4two} and its implied scaling Eq.~\eqref{eq:CMI_equals_G4two_scaling}.

We can bound the local CMI by the pairwise CMI and statistics of loop lengths in the Goldstone phase. Consider a local geometry (Fig.~\ref{fig:cmi_geometries}) where $A$ is a single qubit, and the region $C$ is decomposed into a single qubit $C_1$ at a distance $r$ from $A$ and the rest $C_2 = C \setminus C_1$. By the chain rule of CMI and strong subadditivity, we have 
\begin{align}
    I(A:C|B)_{\text{local}} = I(A:C_1 C_2|B) \geq I(A:C_1|C_2 B),
\end{align}
where the RHS is exactly the pairwise CMI. Thus, in the Goldstone phase, we must have at least as slow a decay as the pairwise CMI i.e. $\log(r/l)^{-12}$.

In the Goldstone phase, Ref.~\cite{Nahum_2013} showed that the non-returning probability $Q(l)$ for a loop through a
fixed link after $l$ steps (their Eq.~(52)) decays as $Q(l)\sim \mathrm{const}/\ln(l/l_0)$. Hence the probability $\pi(r)$ that such
a loop makes an excursion to distance at least $r$ is bounded by $\pi(r)\le Q(r)\lesssim \mathrm{const}/\ln(r/r_0)$.
For the local geometry, a necessary condition for nonzero $I(A{:}C|B)_{\mathrm{local}}$ is that some loop incident on
$A$ reaches $C$ across the buffer, which implies that at least one of the $O(1)$ links in $\partial A$ has an excursion
of size $\ge r$. By a union bound this event has probability $\lesssim \mathrm{const}/\ln(r/r_0)$ and the $I(A{:}C|B)_{\mathrm{local}}$ is upper-bounded by this probability. Putting these bounds together, we infer
\begin{equation}
  I(A{:}C|B)_{\mathrm{local}}(r)\ \sim\ \frac{1}{[\log(r/l)]^{\alpha}},
  \label{eq:localCMI_log_alpha}
\end{equation}
where 
  $1\leq\alpha \leq12$,
up to nonuniversal prefactors and subleading corrections. In particular, the decay is slower than any power law and the Markov length extracted from the local geometry diverges.
Numerically, the local CMI is shown to follow a similar transition as the pairwise CMI, as presented in Fig.~\ref{fig:cplc_cmi_transition}b.

In the \emph{global} geometry, the CMI exhibits a clear transition (for $q=0.1$) at $p_c\simeq 0.571$ with a finite-size scaling collapse consistent with $\nu\simeq 2.79$ [Fig.~\ref{fig:global_cmi_transition}].  For $p<p_c$ the global CMI decays exponentially with $L$, while for $p>p_c$ it grows logarithmically with $L$, as expected in the short-loop and Goldstone regimes, respectively.  In the \emph{local} and \emph{pairwise} geometries, fixing $d_{AC}=L/2$ we find exponential decay with $L$ for $p<p_c$, while for $p\ge p_c$ the decay becomes sub-exponential [Fig.~\ref{fig:local_pair_cmi_scaling}], consistent with the interpretation of pairwise CMI as the $G_4^{(2)}$ component of the four-leg correlator (and hence with the asymptotic scaling inherited from $G_4$).  Along the $p=0$ line, the transition at $q=0.5$ is bond percolation with $\nu=4/3$~\cite{Nahum_2013}; away from $p=0$ the transition leaves the percolation universality class, consistent with the exponent extracted from the data in Fig.~\ref{fig:cplc_cmi_transition}.

From Fig.~\ref{fig:cplc_cmi_transition}a we see that the global CMI exhibits the same qualitative behavior across the phase diagram as the loop-model \emph{spanning number}.  Define $N_{\mathrm{span}}(A\!\leftrightarrow\! C)$ to be the number of loop strands that connect $A$ to $C$ through $B$ in a given loop configuration (and let $\langle\cdot\rangle_f$ denote the average over flag patterns, equivalently over loop configurations).  The spanning number is a natural diagnostic of the short-loop versus Goldstone regimes: in the short-loop phase it is exponentially suppressed with separation, while in the Goldstone phase it grows logarithmically, reflecting the proliferation of long wandering strands~\cite{Nahum_2013}.

$N_{\mathrm{span}}$ provides an upper bound to the global CMI.  Each independent $A$--$C$ parity constraint that survives conditioning on $B$ must be supported on at least two $A$--$C$ spanning strands, which implies the upper bound
\(
  I(A{:}C\mid B)_{\mathrm{global}}
  \;\le\;
  \frac{1}{2}\,\big\langle N_{\mathrm{span}}(A\!\leftrightarrow\! C)\big\rangle_f.
\)
In the short-loop phase, where both quantities decay exponentially in the separation $r$, this immediately implies that the Markov length $\xi_{\mathrm{M}}$ extracted from $I(A{:}C\mid B)_{\mathrm{global}}\sim e^{-r/\xi_{\mathrm{M}}}$ is upper bounded by the loop length scale $\xi_{\mathrm{loop}}$ extracted from $\langle N_{\mathrm{span}}(A\!\leftrightarrow\! C)\rangle_f\sim e^{-r/\xi_{\mathrm{loop}}}$, i.e.\ $\xi_{\mathrm{M}}\le \xi_{\mathrm{loop}}$.

Remarkably, our numerics indicate that this bound is saturated throughout the short-loop phase: the decay lengths extracted from the global CMI and from the spanning number coincide within numerical accuracy, $\xi_{\mathrm{M}}=\xi_{\mathrm{loop}}$ for $p<p_c$, as shown in Fig.~\ref{fig:classical_markov_length}.  This provides direct evidence that the Markov length of the classical ensemble is set by the same loop correlation length controlling strand-spanning in the CPLC model.

\begin{figure}[h]
   \includegraphics[width=0.8\linewidth]{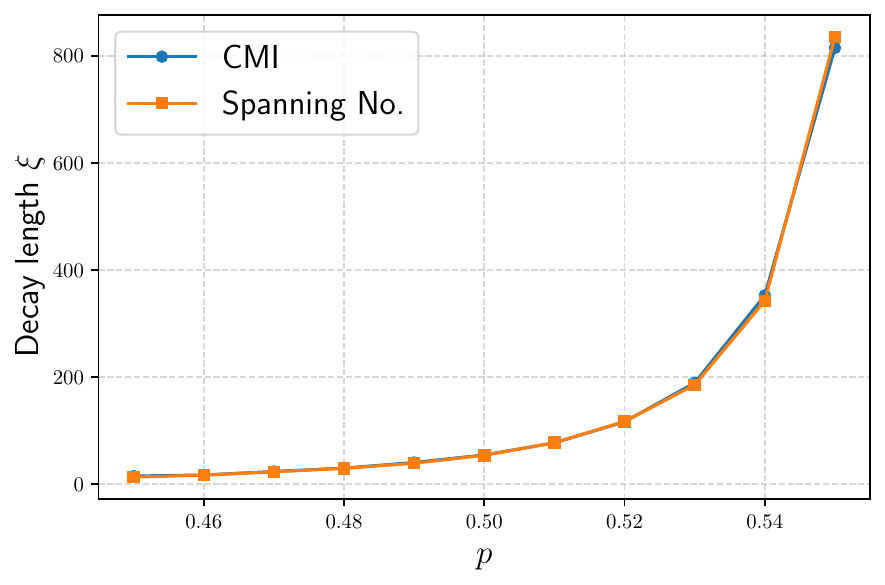}
   \caption{\justifying In the short-loop phase $p<p_c$, the decay length extracted from the global CMI (the Markov length $\xi_{\mathrm{M}}$) coincides within numerical accuracy with the loop length scale $\xi_{\mathrm{loop}}$ extracted from the spanning number.}
   \label{fig:classical_markov_length}
\end{figure}

Combining (i) the numerical scaling of the global CMI across the transition in Fig.~\ref{fig:cplc_cmi_transition}a with (ii) the numerical observation that these two decay lengths coincide throughout the short-loop regime, $\xi_{\mathrm{M}}=\xi_{\mathrm{loop}}$ (Fig.~\ref{fig:classical_markov_length}), we arrive at the scaling forms
\begin{equation}\label{eq:global_spanning}
I(A{:}C\,|\,B)_{\text{global}}\sim
\begin{cases}
\text{constant}, & \text{transition},\\[3pt]
e^{-r/\xi_{\mathrm{loop}}}, & \text{short loop},
\\[3pt]
\log(r/l), & \text{Goldstone}.
\end{cases}
\end{equation}

Our numerical results confirm that the CMI behavior across all geometries follows the CPLC phase diagram in Fig.~\ref{fig:cplc_phase_diagram}. In the short loop phase, the CMI in all geometries are exponentially decaying, implying the existence of a finite Markov length, while at the critical point and the Goldstone phase the Markov length diverges. We quantitatively do not find the Markov length obtained from the local, pairwise, and global geometries to be equal; however, that may arise from the fact that the different geometries could come with distinct $\text{poly}(r)$ factors as well.

In summary, for the classical ensemble $\rho_{p,q}$, we find that the CMI behavior mirrors the loop spanning behavior of the underlying loop model, which exhibits stable critical and short-loop phases and a transition between them. \footnote{We have also studied the CMI and other informational quantities along the line $q = 0.0$, as a function of $p$, and find an apparent crossover at $p^* < 1.0$. As reported in~\cite{Nahum_2013}, the localization length diverges near the $p = 1.0$ fixed point, which leads to strong finite size constraints in the numerical analysis. We report the results in the Appendix~\ref{app:SRE-q0-metallic}, and perform finite size scaling to show that this crossing is not a real transition.}

\section{Critical phase as a memory}
\label{sec:fault_tolerance}

In this section we ask what the extended critical mixed-state phase implies for fault tolerance.
From error correction perspective there are two distinct questions:
(i) whether logical information is still present after the noise, and
(ii) if it is present, whether it can be recovered by a quasi-local decoder or only by intrinsically global processing.
To address (i), we use the mutual information between a reference and the noisy code block, \(I(R{:}Q)\), which plays the role of an input output information measure for the noise channel.
To address (ii), we refine this diagnostic by introducing a ``punctured'' version, obtained by discarding a bulk region of the code before computing the mutual information.
Comparing the full and punctured quantities provides a direct operational probe of \emph{local recoverability}: if decoding can be implemented by a quasi-local circuit, then removing a sufficiently small interior region should not affect the information accessible from the remaining degrees of freedom, whereas a puncture-induced drop signals that any successful decoder must be nonlocal.
We now formalize these notions in the classical setting relevant for storing a logical bit.

\subsection{Memory and the recovery}
\label{subsec:classical-memory-puncture}

Complete \(X\)-dephasing of the surface-code ground state produces a classical ensemble stabilized by the \(X\)-type plaquette checks.
In this ensemble the only remaining logical degree of freedom is the logical operator \(\bar X\), so the code stores a single classical bit.
We couple this bit to a classical reference register \(R\) by enforcing $X_r\bar X = 1$,
without imposing any \(Z\)-type constraint. The corresponding reference--code output state is defined by
\begin{equation}
  \rho_{RQ} := (\mathrm{id}_R \otimes \mathcal D_X)\!\left[
    \left( \frac{1 + X_r\bar X}{2} \right)
    \bigl( \mathbb{I}_R \otimes \rho_{\mathrm{code}} \bigr) 
  \right],
\end{equation}
where \(\mathcal D_X\) completely dephases the code in the \(X\) basis and \(\rho_{\mathrm{code}}\) is chosen to have fully mixed logical space.
The marginal \(\rho_Q\) is the classical surface-code ensemble storing the logical \(X\) bit.

The corresponding noisy reference--code output state is defined, for parameters \((p,q)\), by
\begin{equation}
  \rho_{RQ}(p,q)
  := (\mathrm{id}_R \otimes \mathcal N_{p,q}^{\otimes n})\!\left[
    \left( \frac{1 + X_r\bar X}{2} \right)
    \bigl( \mathbb{I}_R \otimes \rho_{\mathrm{code}} \bigr)
  \right],
\end{equation}

The Pauli dephasing channel \(\mathcal N_{p,q}^{\otimes n}\) acts on \(Q\) to produce the classical ensemble \(\rho_{p,q}\) defined in Eq.~\eqref{eq:state-classical}.
Since \(\mathcal N_{p,q}^{\otimes n}\) outputs both the data qubits and a classical flag register \(F\) indicating the local Pauli basis, we will henceforth include the flags in the output system and write \(Q \equiv Q_{\mathrm{data}}\otimes F\).
Equivalently, the noise channel can be viewed as performing nonselective single-qubit Pauli measurements, with the flags revealing when a \(Y\)- or \(Z\)-basis measurement occurred instead of an \(X\)-basis measurement.

We say that the classical logical bit is recoverable if there exists a recovery channel that can map the noisy state to the clean state: \(\mathcal R\) on \(Q\) such that
\begin{equation}
  (\mathrm{id}_R \otimes \mathcal R)
  \bigl(
    \rho_{RQ}(p,q)
  \bigr)
  \;\to\;
  \rho_{RQ}(0,0).
  \label{eq:local-perfect-recovery-classical}
\end{equation}

A finer question is whether recovery can be implemented by a \emph{local} decoder.
To make this precise, we consider recovery maps \(\mathcal R:Q\to Q\) acting on the physical degrees of freedom equipped with the lattice metric. We say that \(\mathcal R\) is \emph{local} if it can be implemented as a channel circuit of bounded depth, where each layer is a tensor product of CPTP maps supported on regions of bounded diameter (equivalently, contained in the neighborhood of some site).
We will be interested both in \emph{strictly local} recovery, where the circuit depth and gate diameter are \(O(1)\), and in \emph{quasi-local} recovery, where they are allowed to scale polylogarithmically with the linear system size \(L\), i.e.,
\(
  O\bigl((\log L)^k\bigr)
\)
for some fixed \(k\).
This locality constraint directly probes the circuit complexity of decoding and is therefore a sharper notion than global recoverability.

Our notion of local recovery map is closely related to the \emph{local reversibility} framework developed in Ref.~\cite{Sang_2025}.
There, one considers a setting in which a state is generated by evolution under a local Lindbladian.
Ref.~\cite{Sang_2025} shows that if the Markov length (defined via exponential decay of the CMI for annular partitions) remains finite along the Lindbladian trajectory, then the evolution is reversible by another quasi-local Lindbladian, with the locality scale controlled by the Markov length along the path. Equivalently, finite Markov length along the noise evolution provides a \emph{sufficient} condition for the existence of a quasi-local recovery/decoder.

\subsection{Diagnostics of global recoverability}
\label{subsec:global-recoverability}

We first focus on \emph{global} recoverability of the stored classical bit.
The natural diagnostic is the input--output mutual information~\footnote{This is the coherent information $I_{c}(R\rangle Q)\coloneq S(Q)-S(RQ)$ plus a constant equal to the number of logical qubits; therefore, we use coherent and mutual information interchangeably in this context.}  \(I(R{:}Q)\)~\cite{Schumacher_1996,schumacher2001approximatequantumerrorcorrection,Horodecki_2005,Huang_2025, Lee_2025}, which is a standard decoder-independent measure of how much information about the logical input can still be inferred from the noisy code and, equivalently, whether the output retains correlations with the logical reference.
If \(I(R{:}Q)\) saturates to its maximal value of one bit, the logical information is perfectly preserved and hence globally recoverable; if \(I(R{:}Q)\) vanishes in the thermodynamic limit of the system, then the logical information has been lost to the environment and cannot be recovered by any decoder.
Concretely, after the flagged Pauli-measurement noise, \(I(R{:}Q)\) quantifies how much information about the logical input remains accessible from the \emph{entire} output, assuming an ideal decoder with unrestricted access to all degrees of freedom:
\begin{equation}
  I(R{:}Q)=S(R)+S(Q)-S(RQ),
\end{equation}
with logarithms in base \(2\).
Here we take \(Q\equiv Q_{\mathrm{data}}\otimes F\), i.e.\ the output includes both the data qubits and the classical flags.

The flagged noise is a classical mixture of nonselective single-qubit Pauli measurements, so each flag sample specifies a full measurement trajectory \(f\) (which basis \(\sigma_i\in\{X,Y,Z\}\) was applied at every site).
For fixed \(f\), the trajectory induces an effective logical non-selective measurement \(\bar\sigma(f)\in\{\bar X,\bar Y,\bar Z\}\).
Since the reference is correlated only with the logical \(\bar X\) bit, the trajectory mutual information is binary: \(I(R{:}Q)_f=1\) if \(\bar\sigma(f)=\bar X\) and \(I(R{:}Q)_f=0\) otherwise.
Averaging over trajectories gives
\begin{equation}
  I(R{:}Q)=\sum_f p(f)\,I(R{:}Q)_f=\Pr\bigl(\bar\sigma(f)=\bar X\bigr).
  \label{eq:MI-prob}
\end{equation}

To evaluate the right-hand side, we use the loop-model mapping.
A fully measured bulk induces a pairing of the four dangling boundary Majorana modes of the surface-code patch.
There are three topologically distinct pairings, and they correspond to effective logical measurement bases \(Y\), \(X\), and \(Z\), as shown in Fig.~\ref{fig:boundary-pairings}.
Only the \(X\)-type pairing preserves the encoded classical \(\bar X\) bit.
Therefore, Eq.~\eqref{eq:MI-prob} can be computed by sampling trajectories and recording which boundary pairing sector they induce; \(I(R{:}Q)\) is the frequency of the \(X\)-type pairing.

The resulting \(I(R{:}Q)\) obtained from Monte Carlo sampling is shown in Fig.~\ref{fig:Coh_Info_transition}.
It exhibits three regimes, \(I(R{:}Q)\to 1\), \(I(R{:}Q)\to 0\), and \(I(R{:}Q)\to 1/3\); respectively, these coincide with the two short-loop phases and the intervening Goldstone phase of the loop model discussed in Sec.~\ref{sec:classical_ensemble_cmi}.
In one short-loop phase the boundary sector is pinned to the \(X\)-type pairing, so typical trajectories measure \(\bar X\) and \(I(R{:}Q)\to 1\).
In the other short-loop phase the boundary sector is pinned to the \(Z\)-type pairing, so typical trajectories measure \(\bar Z\) and the classical \(\bar X\) information is erased, giving \(I(R{:}Q)\to 0\).
In the Goldstone phase the boundary mode delocalizes into the bulk and the boundary sector is not pinned; in the large-distance limit the three topological pairing sectors occur with equal probability, which yields \(I(R{:}Q)\to 1/3\) via Eq.~\eqref{eq:MI-prob}.

It is useful to translate the mutual information \(I(R{:}Q)\) into an operational decoding statement.
Since the encoded information is a single classical bit, the performance of any decoder can be characterized by its \emph{success probability} \(p_{\mathrm{succ}}\), defined as the probability that the decoder correctly outputs the logical value stored in the reference.
In the usual error-correction setting, one encounters two limiting behaviors in the thermodynamic limit: if the logical information is perfectly preserved, then \(I(R{:}Q)\to 1\) and an optimal decoder succeeds with probability \(p_{\mathrm{succ}}\to 1\); if the logical information is completely lost, then \(I(R{:}Q)\to 0\) and the best any decoder can do is random guessing, which succeeds with probability \(p_{\mathrm{succ}}=1/2\) for a classical bit.
In both cases, increasing the code distance drives the success probability exponentially close to these limiting values.

\begin{figure}[t]
    \centering
    \begin{tikzpicture}[line cap=round, line join=round]
      \def\s{1.2}    
      \def\o{0.2}    

      \newcommand{\squarewithticks}{
        \draw[gray!70,thick] (0,0) rectangle (\s,\s);
        \draw[gray!70,thick] (0,0) -- (-\o,-\o);
        \draw[gray!70,thick] (\s,0) -- (\s+\o,-\o);
        \draw[gray!70,thick] (\s,\s) -- (\s+\o,\s+\o);
        \draw[gray!70,thick] (0,\s) -- (-\o,\s+\o);
      }

      \begin{scope}[xshift=3cm]
        \squarewithticks
        \draw (0,0) .. controls (\s*0.15,\s*0.35) and (\s*0.25,\s*0.55) .. (\s*0.45,\s*0.6);
        \draw (\s*0.45,\s*0.6) .. controls (\s*0.65,\s*0.7) and (\s*0.8,\s*0.9) .. (\s,\s);
        \draw (\s,0) .. controls (\s*0.8,\s*0.15) and (\s*0.6,\s*0.35) .. (\s*0.4,\s*0.55);
        \draw (\s*0.4,\s*0.55) .. controls (\s*0.25,\s*0.75) and (\s*0.1,\s*0.9) .. (0,\s);
      \end{scope}

      \begin{scope}[xshift=0cm]
        \squarewithticks
        \draw (0,0) .. controls (\s*0.2,\s*0.3) and (\s*0.4,\s*0.35) .. (\s*0.6,\s*0.2);
        \draw (\s*0.6,\s*0.2) .. controls (\s*0.8,\s*0.15) and (\s*0.95,\s*0.05) .. (\s,0);
        \draw (0,\s) .. controls (\s*0.25,\s*0.85) and (\s*0.45,\s*0.7) .. (\s*0.6,\s*0.75);
        \draw (\s*0.6,\s*0.75) .. controls (\s*0.8,\s*0.85) and (\s*0.95,\s*0.95) .. (\s,\s);
      \end{scope}

      \begin{scope}[xshift=6cm]
        \squarewithticks
        \draw (0,0) .. controls (\s*0.35,\s*0.15) and (\s*0.4,\s*0.4) .. (\s*0.25,\s*0.6);
        \draw (\s*0.25,\s*0.6) .. controls (\s*0.15,\s*0.8) and (\s*0.05,\s*0.95) .. (0,\s);
        \draw (\s,0) .. controls (\s*0.7,\s*0.2) and (\s*0.65,\s*0.45) .. (\s*0.8,\s*0.65);
        \draw (\s*0.8,\s*0.65) .. controls (\s*0.9,\s*0.85) and (\s*0.98,\s*0.98) .. (\s,\s);
      \end{scope}
    \end{tikzpicture}
    \caption{\justifying Three topologically distinct pairings of the four boundary Majorana modes induced by a trajectory, corresponding to effective logical measurement bases \(X\), \(Y\), and \(Z\) (left to right). Only the \(X\)-type pairing preserves the encoded classical \(\bar X\) bit.}
    \label{fig:boundary-pairings}
\end{figure}
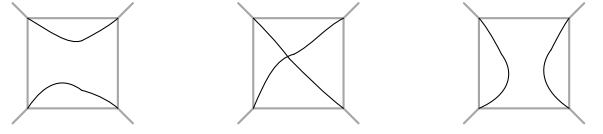

\begin{figure}[t]
    \centering
    \includegraphics[width=\linewidth]{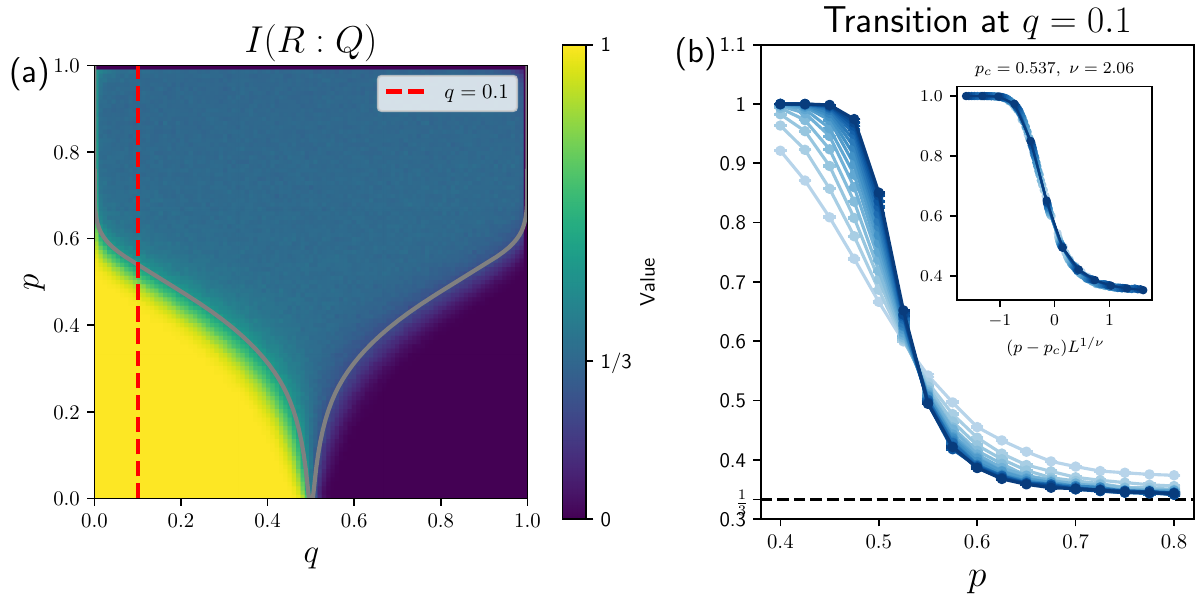}
    \caption{\justifying
    (a) Heat map of the input--output mutual information \(I(R{:}Q)\) in the \((p,q)\) plane for a surface code patch of linear size $L = 101$, computed numerically by sampling flagged measurement trajectories and identifying the induced boundary pairing sector as in Fig.~\ref{fig:boundary-pairings}. The three plateaus \(I(R{:}Q)\to 1\), \(1/3\), and \(0\) correspond to the short-loop--\(X\), Goldstone, and short-loop--\(Z\) regimes of the loop model with crossings.
    Note, the grey line is a guide to the eye and not a phase boundary, rather a contour line at $I(R:Q) = 0.5$.(b) One-dimensional cut at fixed \(q=0.1\), showing \(I(R{:}Q)\) versus \(p\) for several system sizes $L = 21-161$ with $10^4$ samples, together with a finite-size scaling collapse.
    The crossing and collapse locate the transition at \(p_c\) and give a critical exponent \(\nu\simeq 2\).
    }

    \label{fig:Coh_Info_transition}
\end{figure}

The Goldstone phase realizes a qualitatively different regime.
Here the logical information is neither fully preserved nor completely erased: instead, the flagged measurement trajectories induce the three logical sectors \(\{\bar X,\bar Y,\bar Z\}\) with equal probability.
With probability \(1/3\) the trajectory lies in the \(\bar X\) sector, in which case a global decoder has full access to the logical bit and succeeds with unit probability.
With the remaining probability \(2/3\), the trajectory lies in a \(\bar Y\) or \(\bar Z\) sector, for which the output is independent of the encoded \(\bar X\) bit and the optimal decoder can do no better than random guessing, succeeding with probability \(1/2\).
As a result, the optimal decoding success probability in the Goldstone phase is
$
p_{\mathrm{succ}}^{\star}
=\frac{1}{3}\cdot 1+\frac{2}{3}\cdot\frac{1}{2}
=\frac{2}{3}
$.
Thus a decoder can only recover with a fixed probability strictly between the two conventional extremes \(1\) and \(1/2\). In the next section we ask whether any local decoder can achieve an asymptotic success probability larger than \(1/2\).
We answer this question in the negative, and show that no quasi-local decoder can outperform random guessing, despite the existence of a global decoder with success probability \(2/3\).

\subsection{Diagnostics of local recoverability}
\label{subsec:diagnostics-local-recoverability}

We now address a finer question than global recoverability: can the logical information be recovered by a decoder of
bounded circuit complexity, i.e., by a local or quasi-local quantum circuit as previously defined?  In general, there is no known necessary and
sufficient information-theoretic criterion for quasi-local recoverability.  Here we introduce a \emph{necessary} condition for quasi-local recovery, and we will use its contrapositive to certify an intrinsically nonlocal recovery regime in our setting.

The key intuition is that a quasi-local recovery map cannot have arbitrarily long-distance causal influence. In particular, if one removes (``punctures'') a region $Q_p$ before decoding, then a quasi-local decoder can only propagate the effect of this removal to a slightly enlarged neighborhood of the puncture (its causal cone). If this enlarged neighborhood is sufficiently smaller than the code distance (specifically, it cannot support any nontrivial logical operator), then puncturing cannot change the asymptotic reference--output correlations that are accessible from the complement. This motivates the following diagnostic. For each lattice size $L$, choose a region $Q_p(L)\subseteq Q$ (a ``puncture'') and let $Q_{p^c}(L):=Q\setminus Q_p(L)$ denote its complement. Given any state $\rho_{RQ}$, define the punctured mutual information by
\begin{equation}
  I(R{:}Q_{p^c}) := S(R)+S(Q_{p^c})-S(RQ_{p^c}),
  \label{eq:punctured_MI_def}
\end{equation}
where $RQ_{p^c}$ is the reduced state obtained from $\rho_{RQ}$ by tracing out $Q_p$. The puncturing test compares the mutual information between $R$ and the full register $Q$ to the mutual information retained after discarding $Q_p$. The following theorem shows that if recovery can be implemented quasi-locally, then there exist punctures $Q_p(L)$ which grow with system size but nevertheless do not affect the asymptotic reference--output correlations, so that the difference between $I(R{:}Q)$ and $I(R{:}Q_{p^c}(L))$ vanishes in the large-$L$ limit. Before stating the theorem, we define \emph{topological stabilizer codes} as families of stabilizer codes on a $D$-dimensional lattice with geometrically local stabilizers (bounded support diameter) and distance $d(L)\to\infty$ as $L\to\infty$~\cite{Bravyi_2013}.

\begin{theorem}[Punctured coherent/mutual information]
\label{thm:puncturing_test}
Let $\{\rho^{(L)}_{RQ}\}_{L}$ be a family of reference--code states on $R\otimes Q$, where $Q=Q(L)$ is the physical register
of a family of topological stabilizer codes with distance $d=d(L)$.
Let $\rho^{(L)}_{RQ,\mathrm{noisy}}$ be the corresponding noisy reference--output state and let
$\rho^{(L)}_{RQ,\mathrm{ideal}}$ be an ideal target state satisfying all stabilizer checks on $Q$.

If (i) the noisy and ideal families satisfy the conditions \footnote{If the noisy state is obtained by applying a noise channel acting only on $Q$ to the ideal state, then
\eqref{eq:ql_recovery_condition} implies \eqref{eq:MI_match_condition}. We state \eqref{eq:MI_match_condition} separately to allow more general settings in which only part of the reference--output correlations are recoverable, so that the ``ideal'' state represents the best performance achievable by any decoder. This is particularly relevant in the Goldstone phase, where the output retains only partial logical information.}

\begin{align}
  \bigl| I(R{:}Q)_{\rho^{(L)}_{RQ,\mathrm{noisy}}} - I(R{:}Q)_{\rho^{(L)}_{RQ,\mathrm{ideal}}} \bigr|
  &\xrightarrow[L\to\infty]{} 0,
  \label{eq:MI_match_condition}
  \\
  \bigl\|(\mathrm{id}_R\!\otimes\! \mathcal{R}_L)\bigl(\rho^{(L)}_{RQ,\mathrm{noisy}}\bigr)-\rho^{(L)}_{RQ,\mathrm{ideal}}\bigr\|_1
  &\xrightarrow[L\to\infty]{} 0,
  \label{eq:ql_recovery_condition}
\end{align}
for some family of recovery channels $\{\mathcal{R}_L\}_{L}$ implementable by a quasi-local circuit, and
(ii) there exist constants $c>0$ and $\alpha>0$ such that $d(L)\ge cL^\alpha$ for all sufficiently large $L$,
then there exists a family of punctures $\{Q_p(L)\subseteq Q\}_{L}$ whose size $|Q_p(L)|$ grows polynomially in $L$ such that,
writing $Q_{p^c}(L):=Q\setminus Q_p(L)$,
\begin{equation}
  \bigl|
  I(R{:}Q_{p^c})_{\rho^{(L)}_{RQ,\mathrm{noisy}}}
  -
  I(R{:}Q)_{\rho^{(L)}_{RQ,\mathrm{noisy}}}
  \bigr|
  \xrightarrow[L\to\infty]{} 0.
  \label{eq:puncture_invariance_thm}
\end{equation}
\end{theorem}

\begin{proof}[Proof sketch]
A region $M$ is correctible iff there are no non-trivial logical Pauli supported solely on $M$. In particular, any $M$ that cannot support a nontrivial logical operator (e.g.\ any region of weight $<d$) is correctable by definition of the distance. On a $D$-dimensional lattice, a hypercubic region of linear size $\ell$ contains $O(\ell^D)$ qubits, hence it is correctable whenever $\ell < d^{1/D}$. Using $d(L)\ge cL^\alpha$, we may therefore choose punctures $Q_p(L)$ of linear size
\(
\mathrm{diam}(Q_p(L))=\gamma\,L^{\alpha/D}
\)
for some constant $0<\gamma<1$. Quasi-locality of $\mathcal{R}_L$ implies a bounded light-cone: the effect of removing $Q_p(L)$ can only influence the recovered state within a slightly enlarged neighborhood of $Q_p(L)$ (with at most polylogarithmic overhead in linear size). For $\gamma$ chosen sufficiently small, this enlarged neighborhood still has size below the distance scale and is therefore correctable. Hence puncturing can only affect the recovered state inside a correctable region and cannot change the recovered state on its complement. Therefore, the reference correlations are asymptotically unchanged, giving \eqref{eq:puncture_invariance_thm}. See Appendix~\ref{app:puncture-proof} for details.
\end{proof}

The puncture scaling suggested by the proof sketch should be viewed as a general baseline that follows from locality and polynomially growing distance, without invoking any additional structure of the code family. For the surface code, however, one can choose substantially larger punctures: the mutual-information invariance holds for geometric punctures whose linear size is any fixed fraction of the system size. This improvement reflects two standard features of the surface code: correctable regions can be taken with linear size proportional to the distance, and the distance scales linearly with the lattice size. Together, these observations lead to the following corollary.

\begin{corollary}[Surface-code puncture invariance]
\label{cor:surface_uniform_puncture_invariance}
In the setting of Theorem~\ref{thm:puncturing_test}, assume the underlying code family is a surface code on a two-dimensional lattice of linear size $L$. Fix any constant $0<\gamma<1$, and let $Q_p(L)\subseteq Q$ be the puncture defined by a square sublattice (a square ``hole'') of linear size
\begin{equation}
  \mathrm{diam}(Q_p(L)) = \gamma L,
  \qquad 0<\gamma<1.
  \label{eq:surface_puncture_size}
\end{equation}
Writing $Q_{p^c}(L):=Q\setminus Q_p(L)$, the punctured mutual information satisfies $I(R{:}Q_{p^c})_{\rho^{(L)}_{\mathrm{noisy}}} \xrightarrow[L\to\infty]{} I(R{:}Q)_{\rho^{(L)}_{\mathrm{noisy}}}$.
\end{corollary}

\subsubsection{No local recovery in the Goldstone phase}
\label{subsubsec:no-local-recovery-goldstone}

\begin{figure}
    \centering
    \includegraphics[width=\linewidth]{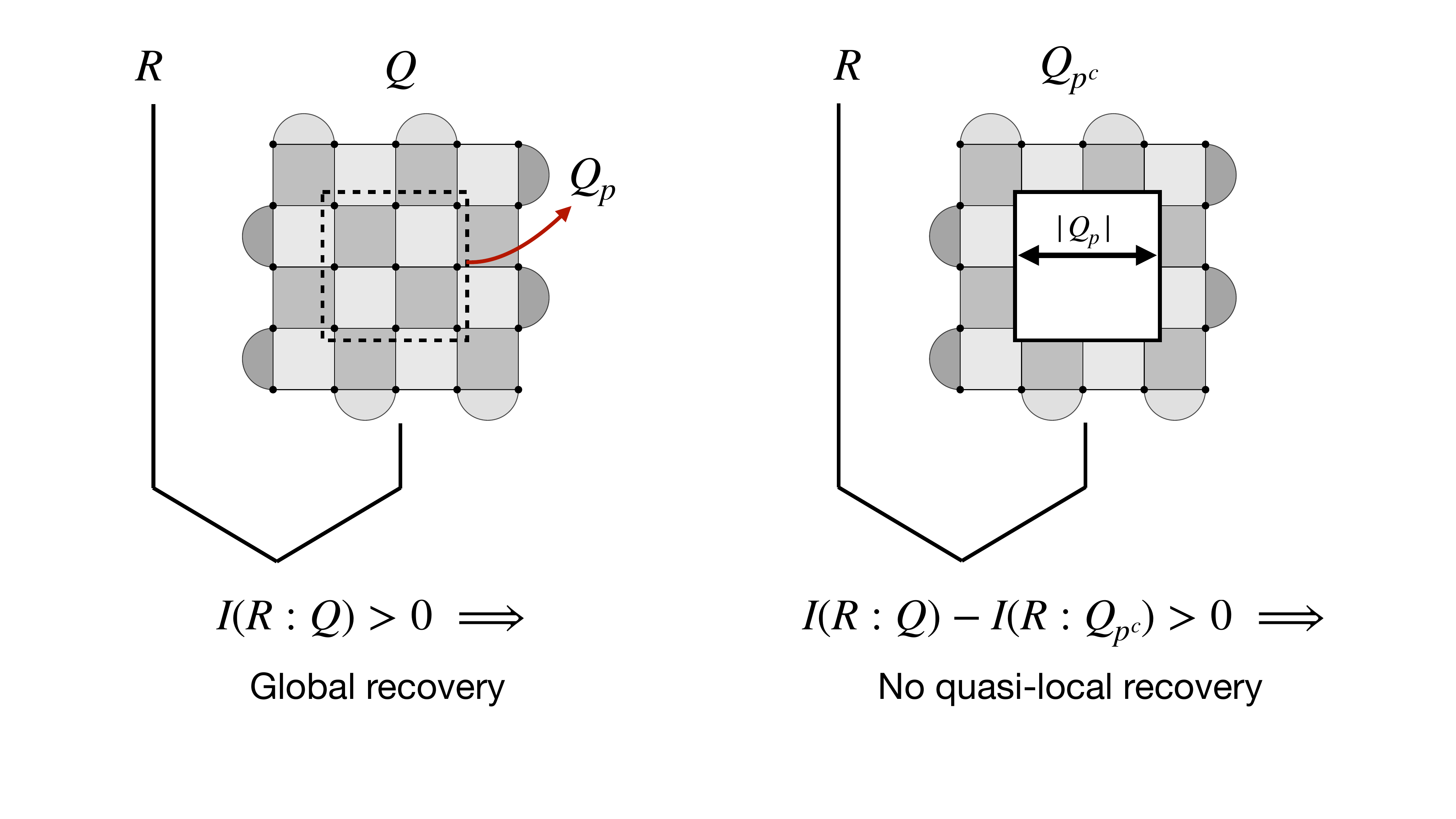}
    \caption{\justifying
    Punctured probes of input--output mutual information. If decoding is quasi-local, discarding a finite fraction of system size does not change the accessible information about the logical input. A persistent drop under puncturing therefore witnesses intrinsically nonlocal recovery.}
    \label{fig:punctured_coh_info}
\end{figure}

\begin{figure}
    \centering
    \includegraphics[width=\linewidth]{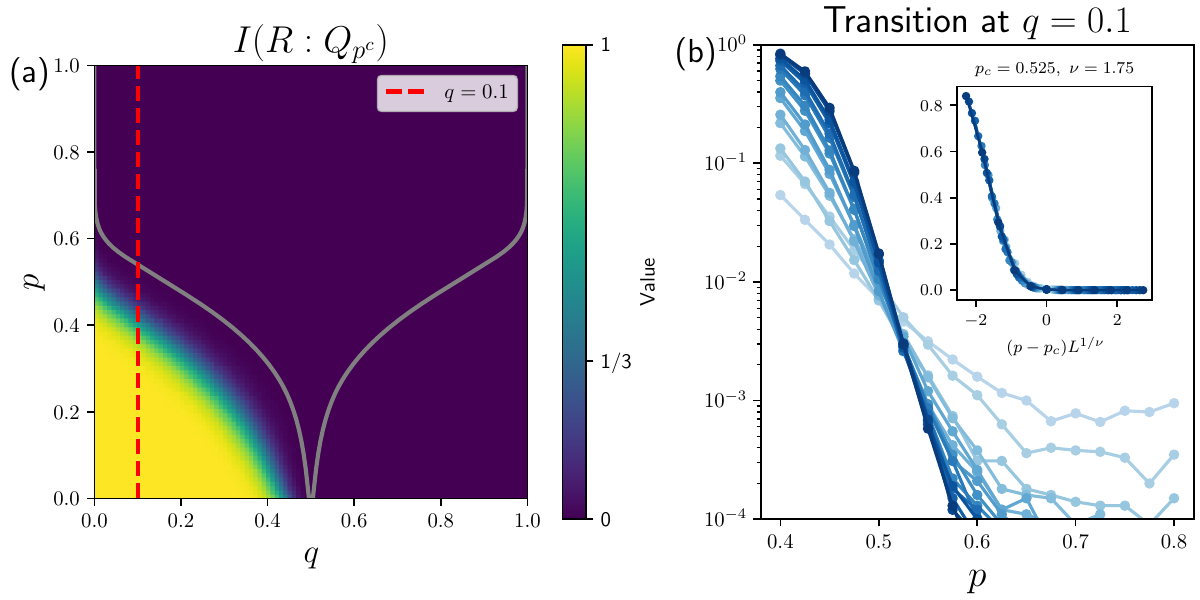}
    \caption{\justifying
    Punctured input--output mutual information \(I(R{:}Q_{p^c})\) computed using the protocol of Fig.~\ref{fig:punctured_coh_info}, with a bulk square puncture of linear size \(L/2\).
    The \(I(R{:}Q)\to 1\) and \(I(R{:}Q)\to 0\) short loop regimes are stable under puncturing, while the \(I(R{:}Q)\to 1/3\) Goldstone regime collapse to \(I(R{:}Q_{p^c})\to 0\).
    The small mismatch between the apparent transition points for punctured and unpunctured curves is consistent with finite-size drift. The plot parameters are the same as Fig.~\ref{fig:Coh_Info_transition}. Note that the purple subregion below the transition arises from the fact that the heatmap is plotted for a fixed system size; finite size scaling reveals that the $I(R:Q_{p^c})$ goes to 1 in the $L \to \infty$ limit below the transition, as shown in (b).
    }
    \label{fig:Punc_Coh_Info_transition}
\end{figure}

We now apply Corollary~\ref{cor:surface_uniform_puncture_invariance} as a diagnostic for nonlocality of decoding in our classical surface-code ensemble.
In the numerics, the puncture \(Q_p\) is a bulk square of linear size \(L/2\), so that a fixed finite fraction of the bulk is discarded.
We then compute \(I(R{:}Q_{p^c})\) by sampling flagged trajectories exactly as for \(I(R{:}Q)\), but tracing out all degrees of freedom (data and flags) supported in \(Q_p\).
Figure~\ref{fig:Punc_Coh_Info_transition} summarizes the results.

The puncturing test cleanly separates the short-loop regimes from the Goldstone phase. In both short-loop phases, puncturing does not change the asymptotic input--output mutual information: in the $\bar X$-pinned short-loop phase we have \(I(R{:}Q)\to 1\) and also \(I(R{:}Q_{p^c})\to 1\), while in the $\bar Z$-pinned short-loop phase we have \(I(R{:}Q)\to 0\) and likewise \(I(R{:}Q_{p^c})\to 0\). In particular, in the $\bar X$-pinned regime, discarding half of the bulk does not affect the recoverability of the bit. This stability admits a simple loop-model interpretation: viable $\bar X$ representatives are confined within distance $\sim \xi_{\mathrm{loop}}$ of the top and bottom boundaries, so the probability that all $\bar X$ representatives intersect a bulk puncture that is a macroscopic distance from the boundary is exponentially small in the boundary-to-puncture separation. Equivalently, the event that puncturing changes the decoder-relevant logical sector is controlled by the same exponentially decaying strand-spanning probabilities that define the short-loop regime.

The Goldstone phase is qualitatively different: although some logical information remains globally present, puncturing destroys all correlations with the reference. Concretely, in the Goldstone regime we find
\begin{equation}
  I(R{:}Q)\to \frac{1}{3}
  \qquad\text{but}\qquad
  I(R{:}Q_{p^c})\to 0,
\end{equation}
so discarding a finite fraction of the bulk eliminates the accessible information about the logical input. By the contrapositive of Corollary~\ref{cor:surface_uniform_puncture_invariance}, this rules out any quasi-local recovery channel in the Goldstone phase.

This collapse is again naturally explained by the loop picture.
In the Goldstone phase the boundary mode delocalizes: strands emanating from the corner Majoranas wander throughout the bulk on all length scales, so the logical sector (which boundary pairing is realized) is encoded in extended bulk correlations rather than boundary-local data.
A simple way to estimate the effect of puncturing is to use the Goldstone scaling of a typical corner strand.
In this phase, the typical strand length grows superlinearly with system size, \(\ell_{\mathrm{strand}}\sim L^{1+x}\) for some \(x>0\) (with \(x=1\) up to logarithmic corrections, consistent with the \(d_f=2\) fractal scaling reported in Ref.~\cite{Nahum_2013}).
Now partition the strand into \(O(L)\)-long coarse-grained segments.
Because the puncture is a fixed macroscopic interior window, each such segment has a finite avoidance probability \(p_a<1\) (equivalently, a finite \(O(1)\) chance to intersect \(Q_p\)).
Since there are \(\sim \ell_{\mathrm{strand}}/L \sim L^{x}\) such segments, the probability that the strand avoids the puncture entirely is at most \(\Pr[\text{avoid }Q_p]\sim p_a^{L^{x}}\to 0\).
Thus a bulk puncture of linear size \(\Theta(L)\) is intersected with probability approaching one by the strands that encode the boundary pairing, and tracing out \(Q_p\) erases that pairing information with probability approaching one, forcing \(I(R{:}Q_{p^c})\to 0\).

Finally, the transition extracted from \(I(R{:}Q_{p^c})\) tracks the loop-model transition, showing that the loss of quasi-local recoverability is controlled by the same phase structure that governs the loop statistics.

\subsection{Decoding transition: global vs local}
\label{subsec:explicit-decoders}

In this subsection we formulate the short-loop--Goldstone transition as a transition in the \emph{circuit complexity} of recovery: while the logical information can remain globally recoverable in an extended parameter regime, the ability to decode using a quasi-local circuit can fail.

For quasi-local recovery there is a known \emph{sufficient} condition based on the Markov length: if the noise can be embedded into an evolution generated by a local Lindbladian and the Markov length remains finite along the evolution (including for partial-layer intermediate states), then the evolution is reversible by a quasi-local recovery map~\cite{Sang_2025}. In addition, we have introduced a \emph{necessary} condition based on the puncturing criterion: a persistent drop of $I(R{:}Q_{p^c})$ under subdistance punctures rules out any quasi-local recovery by the contrapositive of Theorem~\ref{thm:puncturing_test}. Taken together, these criteria strongly constrain when quasi-local decoding is possible.

However, these conditions do not by themselves complete the argument in our setting. The sufficiency theorem of Ref.~\cite{Sang_2025} requires that the channel (or an appropriate path connecting the noisy and reference states) arises from a local Lindbladian evolution, and it is not \emph{a priori} clear that our setting admits such a local Lindbladian connecting $\rho_{0,0}$ to $\rho_{p,q}$. We do not rule out this possibility; therefore, the finiteness of the Markov length should be regarded as suggestive rather than decisive evidence for quasi-local recoverability. For this reason, we construct an explicit quasi-local recovery map to establish that the short-loop phase admits quasi-local decoding.

Combining the explicit quasi-local decoder in the short-loop phase with (i) the optimal global decoder and (ii) the puncturing-based obstruction to quasi-local decoding in the Goldstone phase from the previous section, we conclude that the short-loop--Goldstone transition can be interpreted as a transition in the circuit complexity of recovery.

\begin{figure}[t]
    \centering
    \includegraphics[width=0.7\linewidth]{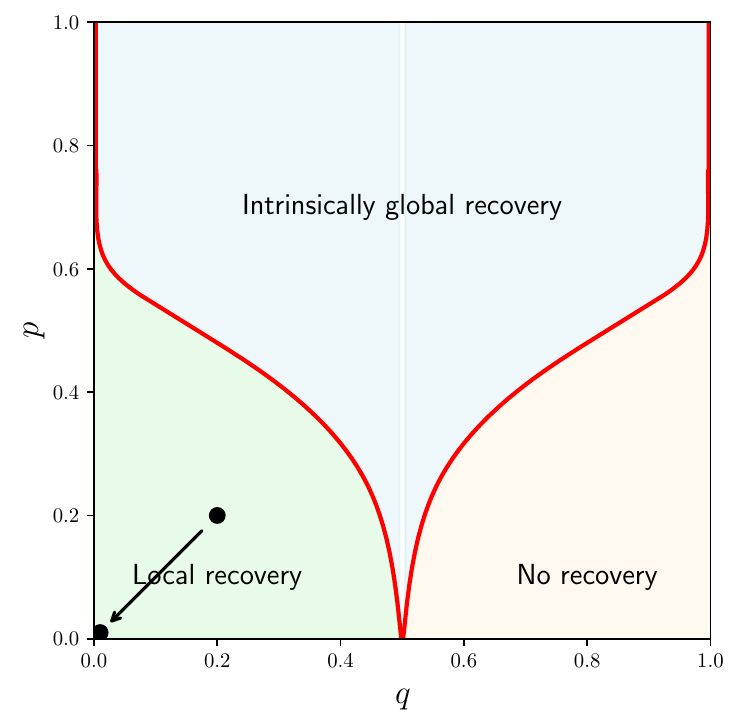}
    \caption{\justifying
    The phase diagram from the perspective of recoverability.
    }
    \label{fig:local_decoder_phase_diagram}
\end{figure}

\subsubsection{Goldstone phase: decoding is intrinsically global}
\label{subsubsec:decoder-goldstone}

\paragraph{Optimal global decoder.}
Because the flagged channel is a classical mixture of nonselective single-qubit Pauli measurements, each flag pattern (trajectory) $f$ induces an effective logical nonselective measurement
\(
\bar\sigma(f)\in\{\bar X,\bar Y,\bar Z\}.
\)
Only trajectories in the $\bar X$ sector can retain any correlation with the reference $R$; when $\bar\sigma(f)\in\{\bar Y,\bar Z\}$ the output is independent of the encoded $\bar X$ bit, so no decoder can do better than random guessing. An optimal decoder therefore has the following form: (i) infer $\bar\sigma(f)$ from the flag record, (ii) if $\bar\sigma(f)\neq \bar X$ output a random guess, and (iii) if $\bar\sigma(f)=\bar X$ apply a recovery that \emph{canonicalizes} the conditional ensemble $\sigma_f$ to a fixed reference-point memory state while preserving the logical $\bar X$ value.

Concretely, conditioned on $\bar\sigma(f)=\bar X$, the flags determine at least one viable $\bar X$ representative; the recovery uses this information to map $\sigma_f$ to the clean memory ensemble (the $p{=}q{=}0$ $\bar X$ memory) while carrying the logical bit along, so that the output has the same correlation with $R$ as the clean memory state. A simple choice for such an operation is to first transfer the logical information from an $\bar X$ representative to one of the qubits, say qubit $k$. (This can be done, e.g., by a measure $\bar X$ -- reset $X_k$ -- conditional $Z_k$ -- discard measurement outcome channel.) We then reset  $X_{j\neq k}\to 1$, transfer the information from $X_k$ to a $q=p=0$ logical $\bar{X}$ (with $k$ not in its support), reset $X_k\to 1$, and dephase all surface-code $Z$-stabilizers. This gives the $q=p=0$ state with the correct logical information.

In the Goldstone phase the logical sector is delocalized, and in the large-system limit the three sectors occur with equal probability, $\Pr[\bar\sigma(f)=\bar X]\to 1/3$. Hence even the globally decoder succeeds only with probability
\(
p^\star_{\mathrm{succ}}=\frac{1}{3}\cdot 1+\frac{2}{3}\cdot\frac{1}{2}=\frac{2}{3},
\)
equivalently $I(R{:}Q)\to 1/3$. The decoder above is optimal in the sense that, whenever correlations with the reference are preserved, it recovers the logical information up to the information-theoretic bound; however, it cannot be implemented by any quasi-local circuit in the Goldstone phase. Operationally, this is exactly what our puncturing diagnostic captures: while the full mutual information remains nonzero, $I(R{:}Q)\to 1/3$, the punctured mutual information vanishes, $I(R{:}Q_{p^c})\to 0$ for subdistance punctures (Sec.~\ref{subsubsec:no-local-recovery-goldstone}), ruling out quasi-local recovery by the contrapositive of Corollary~\ref{cor:surface_uniform_puncture_invariance}.

\subsubsection{Short-loop phase: optimal decoder admits a quasi-local implementation}
\label{subsubsec:decoder-shortloop}

The recovery described above also applies in the short-loop phase: in the $\bar X$-pinned regime one has $\Pr[\bar\sigma(f)=\bar X]\to 1$, so the same optimal strategy succeeds with probability asymptotically approaching $1$. Here, however, our goal is sharper: we aim to establish the existence of a \emph{quasi-local} decoder in the short-loop phase. To do so, we first present a global recovery that is tailored to the short-loop regime and then show how it can be truncated to a quasi-local circuit. 

\paragraph{Global decoder via loop/stabilizer updates.}
Fix a flag configuration $f$ with $\bar\sigma(f)=\bar X$ and write the corresponding conditional ensemble as $\sigma_f$, so that
\(
\rho_{p,q}=\sum_f p(f)\,|f\rangle\!\langle f|\otimes\sigma_f
\)
as in Eq.~\eqref{eq:cq-form}. The recovery step can be organized as a site by site procedure that turns ($Y/Z$) sites one by one in to ($X$) sites. We express this in terms of single-site updates of the form
\begin{equation}
  \sigma_f \longmapsto \sigma_{f^{(i\to X)}},
  \label{eq:single_site_update_maintext}
\end{equation}
where $f^{(i\to X)}$ denotes the flag pattern obtained from $f$ by replacing the local rewiring at site $i$ by the $X$-rewiring, leaving all other flags unchanged. In the loop picture, such an update changes only the reconnection of the four half-edges incident on $i$, so it affects only the one or two loop constraints carried by the strands passing through $i$; all loop-parity constraints supported on loops that avoid $i$ are identical before and after the update. Iterating \eqref{eq:single_site_update_maintext} over all sites yields a global recovery map that drives $\sigma_f$ to the canonical $p=q=0$ ensemble while preserving the logical $\bar X$ value. To show this, it suffices to verify that our CPTP map updates the state \emph{consistently} with the updated loop data: after each local rewiring, the resulting state is stabilized by the stabilizers inferred from the correspondingly updated loops. Since the full stabilizer group for each configuration is generated by its loops, enforcing this consistency at every step guarantees that the overall iteration implements the desired transformation from $\sigma_f$ to $\sigma_{f^{(i\to X)}}$ while maintaining the logical sector.

Here we explain how such a \emph{consistent} loop and stabilizer update works. Fix a site $i$ and consider the four half-edges incident on $i$. Each half-edge is either paired to another half-edge at $i$ as part of a loop connection, or it is connected (through a strand) to a dangling boundary Majorana, i.e., it participates in a logical strand. For the update rule, the only required input is which \emph{pairing} of half-edges is realized: once we identify one connected pair, the remaining two half-edges may be in either of the two possibilities above, and the update proceeds the same way.

Suppose, for definiteness, that two of the half-edges are connected by a loop segment, while the other two are in an a priori unknown configuration (as Fig.~\ref{fig:five-panel}). The local update then replaces the non-canonical $(Y/Z)$ pairing at $i$ by the canonical $X$-pairing. Operationally, this means (i) changing the local rewiring at $i$ and (ii) performing the corresponding stabilizer update so that the set of loop-parity constraints remains consistent. Given $i$, one first identifies the relevant strand connectivity of the incident half-edges, as illustrated in Fig.~\ref{fig:five-panel}$(a_1)$--$(a_3)$. Depending on whether the site is of $Y$-type or $Z$-type and on the sublattice, there are three possible loop patterns. The update replaces these by the appropriate $X$-type pattern, Fig.~\ref{fig:five-panel}$(b_1)$ or $(b_2)$, again depending on the sublattice. In terms of loop constraints, this replacement either \emph{merges} loops (e.g., transitioning from configuration $(a_1)$ to $(b_2)$, where a loop on one side is reconnected into two strands on the other side), or \emph{splits} a loop (e.g., transitioning from $(a_3)$ to $(b_1)$, where a single connection is reconfigured into a loop plus two separate segments).

\begin{figure}[t]
  \captionsetup[subfigure]{labelformat=empty}
  \centering

  \begin{subfigure}[t]{0.14\textwidth}
    \centering
    \includegraphics[width=\linewidth]{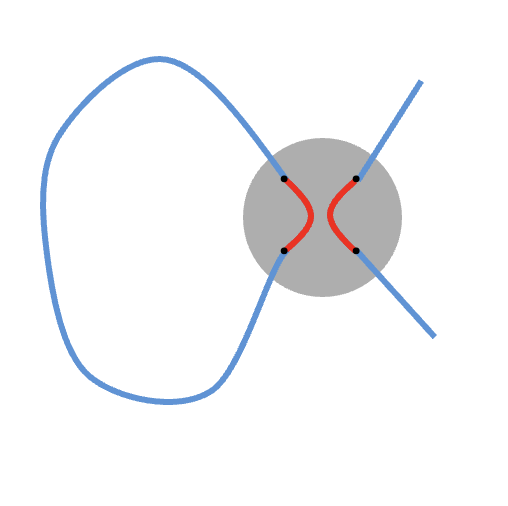}
    \caption{$(a_1)$}
    \label{fig:a1}
  \end{subfigure}\hspace{0.6em}
  \begin{subfigure}[t]{0.14\textwidth}
    \centering
    \includegraphics[width=\linewidth]{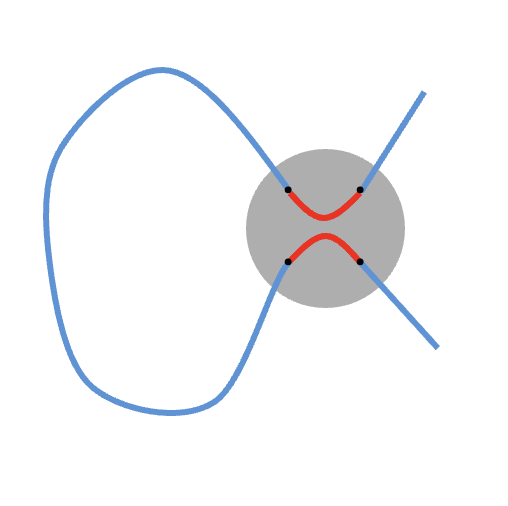}
    \caption{$(a_2)$}
    \label{fig:a2}
  \end{subfigure}\hspace{0.6em}
  \begin{subfigure}[t]{0.14\textwidth}
    \centering
    \includegraphics[width=\linewidth]{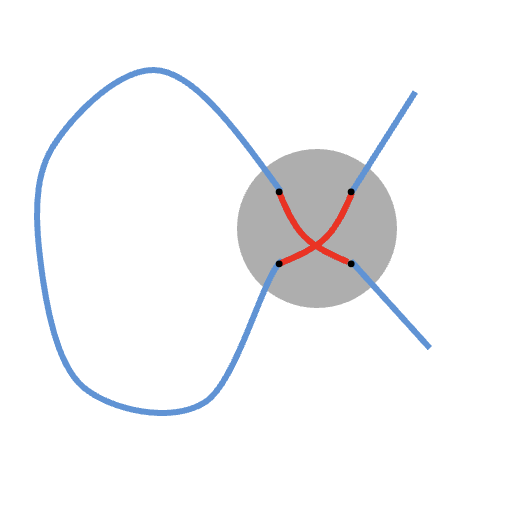}
    \caption{$(a_3)$}
    \label{fig:a3}
  \end{subfigure}

  \vspace{0.6em}

  \begin{subfigure}[t]{0.14\textwidth}
    \centering
    \includegraphics[width=\linewidth]{a1_b1.png}
    \caption{$(b_1)$}
    \label{fig:b1}
  \end{subfigure}\hspace{0.8em}
  \begin{subfigure}[t]{0.14\textwidth}
    \centering
    \includegraphics[width=\linewidth]{a2_b2.png}
    \caption{$(b_2)$}
    \label{fig:b2}
  \end{subfigure}

  \caption{\justifying
    Local loop rewirings used in the single-site update \eqref{eq:single_site_update_maintext}.
    Top row: the three possible non-canonical $(Y/Z)$ routings of the four half-edges incident on a flagged site $i$
    (for the different choices of Pauli type $(Y/Z)$ and sublattice), labeled $(a_1)$-$(a_2)$-$(a_3)$.
    Bottom row: the corresponding canonical $X$-type routings, labeled $(b_1)$ and $(b_2)$.
    Replacing $(a_i)$ by the appropriate $(b_j)$ changes only the local pairing at $i$ and therefore updates only
    the loop-parity constraints on the one or two strands passing through $i$: depending on the incident connectivity,
    the move either splits a strand into two or merges two strands into one, while leaving all constraints supported on loops that avoid $i$ unchanged.
}

  \label{fig:five-panel}
\end{figure}

We now describe how the stabilizer (and $\bar X$-representative) constraints are updated under a single $X_i$-type re-routing. There are two cases.

\emph{Splitting update ($P\mapsto(Q,S)$).}
Assume the loop configuration shows that the flagged qubit $i$ routes two segments of the \emph{same} strand (for instance, a configuration of the type in Fig.~\ref{fig:five-panel}$(a_2)$ or $(a_3)$). Let $P$ be the operator for that strand (either a stabilizer loop operator or a $\bar X$ representative). Since the strand visits $i$ twice, $P$ does not have $i$ in its support. After performing the $X_i$-type re-routing, the single strand is resolved into two strands (for instance, a configuration of the type in Fig.~\ref{fig:five-panel}$(b_1)$) with operators $Q$ and $S$, and both $Q$ and $S$ have site $i$ in their support. These operators satisfy
\(
  P = Q S .
\)
As $P$ corresponds to a single strand, at least one factor, say $S$, is a closed loop. The stabilizer update is implemented by measuring this closed-loop operator $S$. If the outcome is $-1$, apply $Z_i$ so that the post-measurement state is stabilized by $S$ with the right eigenvalue. This introduces $S$ into the stabilizer group. Since $P=QS$ and $P$ was already a valid constraint (either a stabilizer or a $\bar X$ representative), fixing $S$ simultaneously fixes $Q$ via $Q=PS$: if $P$ was a stabilizer then $Q$ is a stabilizer, while if $P$ was a $\bar X$ representative then $Q$ becomes the updated $\bar X$ representative. In this way, the single constraint $P$ is replaced by the pair $(Q,S)$, realizing the splitting move.

\emph{Merging update ($(A,B)\mapsto AB$).}
Assume instead that the loop configuration shows that the flagged qubit $i$ routes segments from \emph{different} strands (for instance, a configuration of the type in Fig.~\ref{fig:five-panel}$(a_1)$). Let $A$ and $B$ be the operators for these two strands (each of $A,B$ is either a stabilizer loop operator or a $\bar X$ representative). Under the $X_i$-type re-routing, the two strands reconnect into a single strand with operator $AB$, and this new strand has no support on site $i$. The desired update therefore keeps $AB$ as the constraint (a stabilizer, or a $\bar X$ representative if exactly one of $A,B$ was one) while removing $A$ and $B$ as separate constraints.

A convenient way to implement this merge at the stabilizer level is by a complete $X_i$ dephasing channel, i.e., a nonselective measurement of $X_i$ (measure $X_i$ and discard the outcome). Because $A$ and $B$ each have the site $i$ in their support, they anticommute with the $X_i$ measurement and are removed by the dephasing. In contrast, the product $AB$ has no support on $i$, so it commutes with the dephasing and remains as a valid constraint. Operationally, this replaces the two constraints $A$ and $B$ by their product $AB$, which is exactly the merging update.

The process explained above assumes either $A$ or $B$ stabilizers are loop like, however there is the exception  when \emph{both} $A$ and $B$ are $\bar X$ representatives. In the short-loop phase this event is exponentially suppressed because viable $\bar X$ representatives are localized near the boundaries; nevertheless, the global decoder can handle it by first moving the logical information onto site $i$ (or another nearby qubit), recording $i$ for later logical retrieval, and then proceeding with the same local merge rule.

\paragraph{Quasi-local truncation in the short-loop phase.}
The intuition behind the truncation is that performing the site-by-site loop deformation requires measuring an appropriate \emph{closed-loop} operator (one of $S$, $A$, or $B$ in the merge/split description above) and applying a local Pauli correction if needed. In the short-loop phase, loops are localized: with high probability the relevant closed-loop segment needed to identify and complete the update is contained in a bounded neighborhood. This suggests that it suffices to access only a local patch to decide the update rule and perform the corresponding measurement/correction.

To make this intuition precise, we explicitly introduce a stability assumption for the short-loop phase under local loop rewiring. Throughout the qubit-by-qubit decoding procedure, the system realizes an inhomogeneous loop ensemble in which some regions have already been converted to the canonical $(p,q)=(0,0)$ ($X$-dephased) connectivity, while the remaining regions are still drawn from parameters $(p,q)$ in the short-loop phase. We assume that such non uniform instances do not induce delocalization: the loop statistics remain exponentially localized. Equivalently, converting a subregion to $(p,q)=(0,0)$ does not generate a growing length scale beyond $\xi_{\mathrm{loop}}$. This assumption is justified since $(p,q)$ and $(0,0)$ lie in the same short-loop phase: a spatially non-uniform instance obtained by converting a subset of regions to $(0,0)$ connectivity is therefore still expected to remain within the same phase, and hence to exhibit short-loop (exponentially localized) statistics controlled by $\xi_{\mathrm{loop}}$.

With this assumption in hand, we now show that the global decoder can be truncated to a quasi-local circuit.
Concretely we consider square tiles of linear size $\ell$. We divide each tile $R$ into a central square $C_R$ of linear size $(1-a)\ell$ and its complement $\bar{C}_R$, an annulus of thickness $a\ell$, to serve as a buffer region. Recovery proceeds as follows: for each \(R\): (i) we read all flags, (ii) perform the above code deformation in $C_R$ if $R$ has the necessary information and abort otherwise. Provided no $R$ is aborted, covering the system requires a constant number of layers: since the $C_R$ from tiling the system cover area $(1-a)^2L^2$, we need $1/(1-a)^2$ layers of tiles to cover the whole system.

The decoder aborts if, for a nonzero flag in $C_R$, no closed-loop segments (i.e., $S$, $A$ or $B$) of the loop-deformation steps lie fully in $R$. The probability of this event is upper bounded by the probability $p_R$ that there exists at least one $\bar{C}_R$-spanning loopstrand. By Markov's inequality, this satisfies
\(
p_R \le \langle N_R\rangle_f,
\)
where $N_R$ is the number of loops spanning $\bar{C}_R$. Here we assume this average quantity has the same statistics everywhere; In particular, loop statistics may differ near the boundary by at most an $O(1)$ factor in the effective loop correlation length $\xi_{\text{loop}}$, and for brevity we ignore this boundary--bulk difference and treat $\xi_{\text{loop}}$ as uniform throughout. Consequently, any bulk bound on $\langle N_R\rangle_f$ applies equally well near the boundary. In the short-loop phase, $\langle N_R\rangle_f$ in the bulk is upper bounded by the spanning number (the number of strands spanning $\bar{C}_R$), since each $\bar{C}_R$-spanning loop contributes at least one $\bar{C}_R$-spanning strand. Using the known behavior of the spanning number in the bulk, we obtain exponential decay in the bulk, and by the same homogeneity assumption we apply the same bound near the boundary. For $a\ell \gg \xi_{\mathrm{loop}}$, we have
\(
\langle N_R\rangle_f \propto \text{poly}(\ell)\exp(-a\ell/\xi_{\mathrm{loop}}).
\)
We can hence estimate \footnote{In addition to the abort mechanism captured by $p_R$, one can consider a distinct logical failure mode in which the lower and upper logical $\bar{X}$ strands both reside within the same tile. In the short-loop phase this contribution is still suppressed as $\exp(-L/\xi_{\text{loop}})$, and is therefore negligible compared to the dominant contribution estimated in Eq~\ref{eq:fail}.}
 how $a\ell$ must scale to suppress the failure rate $P_\text{fail}$ with $L$, i.e., the probability that the decoder aborts for at least one region $R$. By a union bound and up to constant factors,
\begin{equation}
  P_\text{fail} \lesssim L^2 p_R \lesssim L^2\,\text{poly}(\ell)\exp(-a\ell/\xi_{\mathrm{loop}}).
  \label{eq:fail}
\end{equation}
Hence, requiring $a\ell/\xi_{\mathrm{loop}} > \log L^x$ with $x>2$ is sufficient for $P_\text{fail}$ to be suppressed as $L\to\infty$. In particular, choosing $\ell=\polylog(L)$ gives a \emph{polylogarithmic range} decoder (the recovery map acts only within tiles of diameter $O(\ell)$), and the circuit depth is also $\polylog(L)$ because within each tile the loop deformation updates on the sites in $C_R$ must be performed sequentially as the required loop measurements and stabilizer updates depend on the outcomes and on previously updated local pairings.

Finally, we emphasize that the above quasi-local truncation argument is intended only for the $\bar X$-pinned short-loop phase. It does not apply in the other phases, nor does it apply to the rare trajectories (even within the short-loop regime) where the effective logical measurement is $\bar Z$ or $\bar Y$: in those cases global recovery already fails, so one should not expect the quasi-local decoder to succeed. Equivalently, such trajectories are not in the $\bar X$-pinned short-loop sector---they either already contain long strands/loops, or effectively lie in the $\bar Z$-pinned short-loop phase. In either case, inserting the $(p,q)=(0,0)$ connectivity into a patch creates an interface between distinct phases, which generically produces long strands. Consequently, the stability assumption used in the truncation argument breaks down.

\section{Discussion and outlook}
\label{sec:discussion}

We have analyzed an ensemble obtained from the surface code under heralded random Pauli noise and shown that it realizes an extended mixed-state critical phase with short-range local correlations but divergent Markov length, in particular polylogarithmic decay of CMI. In terms of error correction properties, we found using coherent information that the critical phase retains a partial logical memory that is only nonlocally recoverable.  To rule out quasi-local decoding, we introduced a punctured version of coherent information.  In contrast, in the phase corresponding to short loops in which logical memory is preserved, we proposed an explicit quasi-local decoder.

It would be interesting to generalize the construction to other low-density parity check (LDPC) codes and lattice geometries, where the underlying loop or network models can exhibit different symmetry classes and possibly different types of critical phases.  One could explore dynamical versions of this setup, in which coherent noise and measurements act repeatedly, to study the interplay between mixed-state criticality and measurement-induced transitions. Furthermore, it would be useful to understand whether the critical phase in this work is robust to certain perturbations, and relatedly, if it can be understood in terms of strong-to-weak spontaneous symmetry breaking (SWSSB) \cite{LeeJianXu2023QuantumCriticality,lessa_strong--weak_2024, sala2024spontaneousstrongsymmetrybreaking}.   

While the current work focused on critical mixed state phases that are separable, it would be interesting to find an example of a {\it quantum} critical mixed-state phase, characterized not only by sub-exponential decay of CMI but also long-range entanglement.  Verifying this possibility \cite{kitptalk}, and its implications for quantum memory, is a work in progress.

\begin{acknowledgments}
We thank Ehud Altman, Yimu Bao, Sergey Bravyi, Dongjin Lee, Leonardo A.~Lessa, Andreas W.~W.~Ludwig, Shengqi Sang, Simon Trebst, Sagar Vijay, Stephen Yan, and Guo-Yi Zhu for many helpful discussions.


This research was supported by EPSRC grant EP/V062654/1, NSF PHY-2309135 to the Kavli Institute for Theoretical Physics (KITP), the Perimeter Institute for Theoretical Physics (PI), the Natural
Sciences and Engineering Research Council of Canada
(NSERC), and an Ontario Early Researcher Award.  Research at Perimeter Institute is supported in part by
the Government of Canada through the Department of
Innovation, Science and Industry Canada and by the
Province of Ontario through the Ministry of Colleges and
Universities. 
\end{acknowledgments}

\bibliographystyle{apsrev4-2}
\bibliography{ref}

\appendix

\section{Punctured coherent information and local recoverability}
\label{app:puncture-proof}

This Appendix proves Theorem~\ref{thm:puncturing_test}. Fix a family of topological stabilizer codes on a
$D$-dimensional lattice with physical register $Q=Q(L)$ and distance $d(L)$, and let
$\rho^{(L)}_{RQ,\mathrm{noisy}}$ and $\rho^{(L)}_{RQ,\mathrm{ideal}}$ denote the corresponding noisy and target
reference--output states appearing in Theorem~\ref{thm:puncturing_test}. Throughout we assume the hypotheses of
Theorem~\ref{thm:puncturing_test}, namely: (i) the mutual informations match asymptotically,
\begin{equation}
\bigl| I(R{:}Q)_{\rho^{(L)}_{RQ,\mathrm{noisy}}} - I(R{:}Q)_{\rho^{(L)}_{RQ,\mathrm{ideal}}} \bigr|
\xrightarrow[L\to\infty]{} 0 ,
\label{eq:MI_match_condition_app}
\end{equation}
(ii) there exists a family of recovery channels $\{\mathcal R_L\}_L$ implementable by a quasi-local circuit such that
\begin{equation}
\bigl\|(\mathrm{id}_R \otimes \mathcal R_L)\bigl(\rho^{(L)}_{RQ,\mathrm{noisy}}\bigr)
-\rho^{(L)}_{RQ,\mathrm{ideal}}\bigr\|_1
\le \varepsilon_L \xrightarrow[L\to\infty]{} 0 ,
\label{eq:ql_recovery_condition_app}
\end{equation}
and (iii) the code distance obeys $d(L)\ge cL^\alpha$ for some constants $c,\alpha>0$ and all sufficiently large $L$.

The proof has two components. First, quasi-locality of $\mathcal R_L$ induces a bounded light cone: puncturing a region
$Q_p(L)$ can only affect the action of $\mathcal R_L$ inside an enlarged neighborhood $Q_{\mathrm{LC}}(L)$.
Second, using the distance lower bound we choose $Q_p(L)$ so that $Q_{\mathrm{LC}}(L)$ is correctable for the code.
This yields a recovery from the punctured noisy state with error controlled by $\varepsilon_L$ (Lemma~\ref{lem:punctured_recovery_compact}).
Finally, a simple data-processing argument implies that puncturing does not change the mutual
information asymptotically, proving Theorem~\ref{thm:puncturing_test}.

\subsection{Circuit locality and light cones}

We recall the circuit notion of quasi-locality used in the main text. Let $(Q(L),d_{\mathrm{lat}})$ be the lattice metric
space of physical qubits. A CPTP map $\mathcal R_L:Q\to Q$ is a \emph{local channel circuit} of depth $t_L$ and range
$w_L$ if it can be written as
\begin{equation}
\mathcal R_L
=
\mathcal E^{(t_L)}_L \circ \cdots \circ \mathcal E^{(1)}_L ,
\label{eq:channel_circuit_def_app}
\end{equation}
where each layer factorizes over disjoint supports,
\begin{equation}
\mathcal E^{(s)}_L
=
\bigotimes_j \mathcal E^{(s)}_{L,j},
\qquad
\mathrm{diam}\!\bigl(\mathrm{supp}(\mathcal E^{(s)}_{L,j})\bigr)\le w_L .
\label{eq:layer_def_app}
\end{equation}
We say the circuit is \emph{strictly local} if $t_L=O(1)$ and $w_L=O(1)$, and \emph{quasi-local} if
$t_L=\polylog(L)$ and $w_L=\polylog(L)$.

A depth-$t_L$, range-$w_L$ circuit has a light-cone radius
\begin{equation}
r_L := c_0\, t_L w_L ,
\label{eq:rL_def_app_new}
\end{equation}
for a constant $c_0>0$ depending only on geometric conventions; in the quasi-local case $r_L=\polylog(L)$.
Fix a puncture $Q_p(L)\subset Q$ and define the associated light-cone neighborhood
\begin{equation}
Q_{\mathrm{LC}}(L)
:=
\mathsf N_{r_L}\!\bigl(Q_p(L)\bigr),
\qquad
Q_{\mathrm{LC}}^{c}(L)
:=
Q\setminus Q_{\mathrm{LC}}(L),
\label{eq:QLC_def_app_new}
\end{equation}
where $\mathsf N_{r}(\cdot)$ is the metric $r$-neighborhood.

\subsection{Proof of Theorem~\ref{thm:puncturing_test}}

Theorem~\ref{thm:puncturing_test} follows from the next lemma, which upgrades recovery of the unpunctured noisy state to
recovery of the punctured noisy state once the enlarged puncture $Q_{\mathrm{LC}}(L)$ is correctable. Its proof is
given in Sec.~\ref{app:subsec:proof_lem1}.

\begin{lemma}[Punctured recovery]
\label{lem:punctured_recovery_compact}
Assume correctability of the enlarged puncture, i.e.\ there exists a CPTP map
$\mathcal T_L:Q_{\mathrm{LC}}^{c}(L)\to Q(L)$ such that
\begin{equation}
(\mathrm{id}_R \otimes \mathcal T_L)
\Bigl[
(\mathrm{id}_R \otimes \operatorname{tr}_{Q_{\mathrm{LC}}})
\bigl(\rho^{(L)}_{RQ,\mathrm{ideal}}\bigr)
\Bigr]
=
\rho^{(L)}_{RQ,\mathrm{ideal}} .
\label{eq:ideal_correctability_app_new}
\end{equation}
Let $\rho^{(L)}_{RQ_{p^c}}:=\operatorname{tr}_{Q_p(L)}\rho^{(L)}_{RQ,\mathrm{noisy}}$ be the punctured noisy state.
Then there exists a recovery channel $\widetilde{\mathcal R}_L$ acting on $Q_{p^c}(L)$ such that
\begin{equation}
\bigl\|
(\mathrm{id}_R \otimes \widetilde{\mathcal R}_L)\bigl(\rho^{(L)}_{RQ_{p^c}}\bigr)
-
\rho^{(L)}_{RQ,\mathrm{ideal}}
\bigr\|_1
\le \varepsilon_L .
\label{eq:punctured_recovery_final_app_new}
\end{equation}
Moreover, $\widetilde{\mathcal R}_L$ may be chosen as
\begin{align}
\widetilde{\mathcal R}_L
&=
\mathcal T_L
\circ \operatorname{tr}_{Q_{\mathrm{LC}}}
\circ \mathcal R_L
\circ \bigl(\tau_{Q_p(L)} \otimes \mathrm{id}_{Q_{p^c}(L)}\bigr).
\label{eq:punctured_recovery_explicit_app_new}
\end{align}
i.e., append fresh qubits on the puncture in the state $\tau_{Q_p(L)}$, apply the original recovery $\mathcal R_L$ on
$Q$, discard the causal-cone region $Q_{\mathrm{LC}}(L)$, and finally apply the reconstruction $\mathcal T_L$.
\end{lemma}

We now complete the proof of Theorem~\ref{thm:puncturing_test}. By Lemma~\ref{lem:punctured_recovery_compact},
the recovered punctured state $(\mathrm{id}_R \otimes \widetilde{\mathcal R}_L)(\rho^{(L)}_{RQ_{p^c}})$ is
$\varepsilon_L$-close in trace norm to $\rho^{(L)}_{RQ,\mathrm{ideal}}$.
Since the reference system $R$ has fixed dimension, continuity of mutual information implies
\begin{equation}
I(R{:}Q)_{(\mathrm{id}_R \otimes \widetilde{\mathcal R}_L)(\rho^{(L)}_{RQ_{p^c}})}
-
I(R{:}Q)_{\rho^{(L)}_{RQ,\mathrm{ideal}}}
\xrightarrow[L\to\infty]{} 0 .
\label{eq:MI_recovered_punctured_to_ideal_app}
\end{equation}
On the other hand, by the data-processing inequality applied to the channel $\widetilde{\mathcal R}_L$ acting on
$Q_{p^c}(L)$,
\begin{equation}
I(R{:}Q)_{(\mathrm{id}_R \otimes \widetilde{\mathcal R}_L)(\rho^{(L)}_{RQ_{p^c}})}
\le
I(R{:}Q_{p^c}(L))_{\rho^{(L)}_{RQ,\mathrm{noisy}}}.
\label{eq:DPI_1_app}
\end{equation}
Applying data processing again to the partial trace $\operatorname{tr}_{Q_p(L)}$ (a channel on $Q$) yields
\begin{equation}
I(R{:}Q_{p^c}(L))_{\rho^{(L)}_{RQ,\mathrm{noisy}}}
\le
I(R{:}Q)_{\rho^{(L)}_{RQ,\mathrm{noisy}}}.
\label{eq:DPI_2_app}
\end{equation}
Combining \eqref{eq:MI_recovered_punctured_to_ideal_app}--\eqref{eq:DPI_2_app} with the mutual-information matching
condition \eqref{eq:MI_match_condition_app} gives:
\begin{equation}
I(R{:}Q_{p^c}(L))_{\rho^{(L)}_{RQ,\mathrm{noisy}}}
-
I(R{:}Q)_{\rho^{(L)}_{RQ,\mathrm{noisy}}}
\xrightarrow[L\to\infty]{} 0 ,
\label{eq:puncture_invariance_app}
\end{equation}
which is the claim \eqref{eq:puncture_invariance_thm} of Theorem~\ref{thm:puncturing_test}.

We now choose a puncture family $\{Q_p(L)\}_L$ with $|Q_p(L)|$ growing polynomially in $L$ such that the enlarged
neighborhood $Q_{\mathrm{LC}}(L)$ is correctable.

Fix a lattice-dependent constant $\kappa>0$ such that any hypercubic region of linear size $u$ contains at most
$\kappa u^D$ qubits. Choose $Q_p(L)\subseteq Q(L)$ to be a bulk hypercubic region of linear size $\ell_L$.
Since $Q_{\mathrm{LC}}(L)=\mathsf N_{r_L}(Q_p(L))$ is contained in a hypercube of linear size $\ell_L+2r_L$, we have
\begin{equation}
|Q_{\mathrm{LC}}(L)|
\le \kappa \, (\ell_L+2r_L)^D .
\label{eq:QLC_volume_bound_app}
\end{equation}
Because $r_L=\polylog(L)$ and $d(L)\ge cL^\alpha$, we choose $\ell_L$ so that $\kappa(2\ell_L)^D = d(L)$, which implies
$r_L=o(\ell_L)$. In particular, for all sufficiently large $L$,
\begin{equation}
|Q_{\mathrm{LC}}(L)| < d(L) .
\label{eq:QLC_subdistance_app}
\end{equation}
Thus $Q_{\mathrm{LC}}(L)$ is a correctable region for the code family, and hence there exists a reconstruction map
$\mathcal T_L$ satisfying \eqref{eq:ideal_correctability_app_new}.
This establishes the hypotheses of Lemma~\ref{lem:punctured_recovery_compact} and completes the proof of
Theorem~\ref{thm:puncturing_test}.

\subsection{Causal-cone insensitivity}
\label{app:subsec:causal_cone}

We formalize the causal-cone statement used implicitly in Lemma~\ref{lem:punctured_recovery_compact}.

\begin{lemma}[Light-cone causality for channel circuits]
\label{lem:lightcone_causality_app}
Let $\mathcal R_L$ be a local channel circuit of depth $t_L$ and range $w_L$, and let $r_L$ be as in
Eq.~\eqref{eq:rL_def_app_new}. Then for any CPTP map $\Phi$ supported on $Q_p(L)$ and any state $\rho_{RQ}$,
\begin{equation}
\begin{aligned}
&(\mathrm{id}_R \otimes \operatorname{tr}_{Q_{\mathrm{LC}}})
\circ (\mathrm{id}_R \otimes \mathcal R_L)(\rho_{RQ})\\
&\;=\;
(\mathrm{id}_R \otimes \operatorname{tr}_{Q_{\mathrm{LC}}})
\circ (\mathrm{id}_R \otimes \mathcal R_L)
\circ (\mathrm{id}_R \otimes \Phi)(\rho_{RQ}) .
\end{aligned}
\label{eq:causality_app_new}
\end{equation}
\end{lemma}

\begin{proof}
Let $O$ be any observable supported on $R \cup Q_{\mathrm{LC}}^{c}(L)$. Consider the Heisenberg-picture evolution of $O$
under $\mathrm{id}_R\otimes \mathcal R_L^\dagger$. Because each layer $\mathcal E_L^{(s)}$ factorizes into disjoint
gates of diameter at most $w_L$, the support of $O$ can expand by at most $O(w_L)$ in lattice diameter per layer.
After $t_L$ layers, the support of the backwards-evolved observable remains disjoint from $Q_p(L)$ by the definition
$Q_{\mathrm{LC}}(L)=\mathsf N_{r_L}(Q_p(L))$ with $r_L=c_0 t_L w_L$. Hence
\begin{equation}
\begin{aligned}
&\operatorname{tr}\!\left[
O\, (\mathrm{id}_R \otimes \mathcal R_L)(\rho_{RQ})
\right]\\
&=
\operatorname{tr}\!\left[
O\, (\mathrm{id}_R \otimes \mathcal R_L)\bigl((\mathrm{id}_R \otimes \Phi)(\rho_{RQ})\bigr)
\right].
\end{aligned}
\end{equation}

for all such $O$. Since equality of all expectation values on $R\cup Q_{\mathrm{LC}}^{c}(L)$ is equivalent to equality
of the reduced states on $RQ_{\mathrm{LC}}^{c}(L)$, this proves Eq.~\eqref{eq:causality_app_new}.
\end{proof}

\subsection{Proof of Lemma~\ref{lem:punctured_recovery_compact}}
\label{app:subsec:proof_lem1}

\begin{proof}
Fix any product state $\tau_{Q_p(L)}$ on $Q_p(L)$ and define the ``pad'' map
\begin{equation}
\mathcal A_L := \tau_{Q_p(L)} \otimes \mathrm{id}_{Q_{p^c}(L)},
\end{equation}
which maps states on $Q_{p^c}(L)$ to states on $Q(L)$ by adjoining fresh qubits in the puncture.
Define the punctured-and-padded state
\begin{equation}
\widetilde\rho^{(L)}_{RQ}
:=
(\mathrm{id}_R \otimes \mathcal A_L)\bigl(\rho^{(L)}_{RQ_{p^c}}\bigr).
\label{eq:rho_tilde_def_app}
\end{equation}

\paragraph{Step 1: invariance outside the causal cone.}
Apply Lemma~\ref{lem:lightcone_causality_app} with $\Phi$ equal to the CPTP map
$\operatorname{tr}_{Q_p(L)}$ followed by padding with $\tau_{Q_p(L)}$. Equivalently, puncturing-and-padding acts only
on $Q_p(L)$, so it cannot influence the reduced recovered state outside $Q_{\mathrm{LC}}(L)$. Concretely,
\begin{equation}
\begin{aligned}
&(\mathrm{id}_R \otimes \operatorname{tr}_{Q_{\mathrm{LC}}})
\bigl[(\mathrm{id}_R \otimes \mathcal R_L)\bigl(\rho^{(L)}_{RQ,\mathrm{noisy}}\bigr)\bigr]\\
&\;=\;
(\mathrm{id}_R \otimes \operatorname{tr}_{Q_{\mathrm{LC}}})
\bigl[(\mathrm{id}_R \otimes \mathcal R_L)\bigl(\widetilde\rho^{(L)}_{RQ}\bigr)\bigr].
\end{aligned}
\label{eq:ML_invariance_app}
\end{equation}

\paragraph{Step 2: transfer the unpunctured recovery guarantee to the punctured-and-padded input.}
By contractivity of trace distance under partial trace, Eq.~\eqref{eq:ql_recovery_condition_app} implies
\begin{align}
&\Bigl\|
(\mathrm{id}_R \otimes \operatorname{tr}_{Q_{\mathrm{LC}}})
\bigl[(\mathrm{id}_R \otimes \mathcal R_L)\bigl(\rho^{(L)}_{RQ,\mathrm{noisy}}\bigr)\bigr]
\nonumber\\
&\qquad-
(\mathrm{id}_R \otimes \operatorname{tr}_{Q_{\mathrm{LC}}})
\bigl(\rho^{(L)}_{RQ,\mathrm{ideal}}\bigr)
\Bigr\|_1
\le \varepsilon_L .
\label{eq:reduced_close_app}
\end{align}
Using Eq.~\eqref{eq:ML_invariance_app}, we may replace the first term by the punctured-and-padded input:
\begin{align}
&\Bigl\|
(\mathrm{id}_R \otimes \operatorname{tr}_{Q_{\mathrm{LC}}})
\bigl[(\mathrm{id}_R \otimes \mathcal R_L)\bigl(\widetilde\rho^{(L)}_{RQ}\bigr)\bigr]
\nonumber\\
&\qquad-
(\mathrm{id}_R \otimes \operatorname{tr}_{Q_{\mathrm{LC}}})
\bigl(\rho^{(L)}_{RQ,\mathrm{ideal}}\bigr)
\Bigr\|_1
\le \varepsilon_L .
\label{eq:reduced_close_punctured_app}
\end{align}

\paragraph{Step 3: reconstruct the discarded causal cone using correctability.}
Apply the CPTP reconstruction $\mathcal T_L$ (assumed to exist by
Eq.~\eqref{eq:ideal_correctability_app_new}) to both states in Eq.~\eqref{eq:reduced_close_punctured_app}. By
contractivity under CPTP maps,
\begin{align}
&\Bigl\|
(\mathrm{id}_R \otimes \mathcal T_L)
(\mathrm{id}_R \otimes \operatorname{tr}_{Q_{\mathrm{LC}}})
(\mathrm{id}_R \otimes \mathcal R_L)\bigl(\widetilde\rho^{(L)}_{RQ}\bigr)
\nonumber\\
&\qquad-
(\mathrm{id}_R \otimes \mathcal T_L)
(\mathrm{id}_R \otimes \operatorname{tr}_{Q_{\mathrm{LC}}})
\bigl(\rho^{(L)}_{RQ,\mathrm{ideal}}\bigr)
\Bigr\|_1
\le \varepsilon_L .
\label{eq:apply_T_app}
\end{align}
The second term equals $\rho^{(L)}_{RQ,\mathrm{ideal}}$ by Eq.~\eqref{eq:ideal_correctability_app_new}, hence
\begin{align}
\Bigl\|
(\mathrm{id}_R \otimes \mathcal T_L)
(\mathrm{id}_R \otimes \operatorname{tr}_{Q_{\mathrm{LC}}})
(\mathrm{id}_R \otimes &\mathcal R_L)\bigl(\widetilde\rho^{(L)}_{RQ}\bigr)\nonumber\\
&-
\rho^{(L)}_{RQ,\mathrm{ideal}}
\Bigr\|_1
\le \varepsilon_L .
\label{eq:punctured_recovery_error_app}
\end{align}

\paragraph{Step 4: define the punctured recovery map.}
Define a recovery channel on $Q_{p^c}(L)$ by
\begin{equation}
\widetilde{\mathcal R}_L
:=
\mathcal T_L \circ \operatorname{tr}_{Q_{\mathrm{LC}}} \circ \mathcal R_L \circ \mathcal A_L .
\end{equation}
This has the explicit implementation in Eq.~\eqref{eq:punctured_recovery_explicit_app_new}. Applying it to
$\rho^{(L)}_{RQ_{p^c}}$ gives the left-hand side of Eq.~\eqref{eq:punctured_recovery_error_app}, which proves
Eq.~\eqref{eq:punctured_recovery_final_app_new}.
\end{proof}

\section{Apparent critical behaviour at $q=0$}
\label{app:SRE-q0-metallic}

In the mixed state derived from the heralded dephasing channel, along the $q=0$ line, the CMI and the input–output mutual information both show a sharp crossover at a value $p^\ast<1$ that appears nearly independent of system size over the range we simulate. In particular, for $L=21$--$161$ the curves in Fig.~\ref{fig:classical-memory-phase-apparent-transition} exhibit an apparent crossing near $p^\ast \simeq 0.7$, which one might naively interpret as a genuine classical-memory transition at $q=0$. 

\begin{figure}
    \centering

    \begin{subfigure}[b]{1\linewidth}
        \centering
        \includegraphics[width=0.65\linewidth]{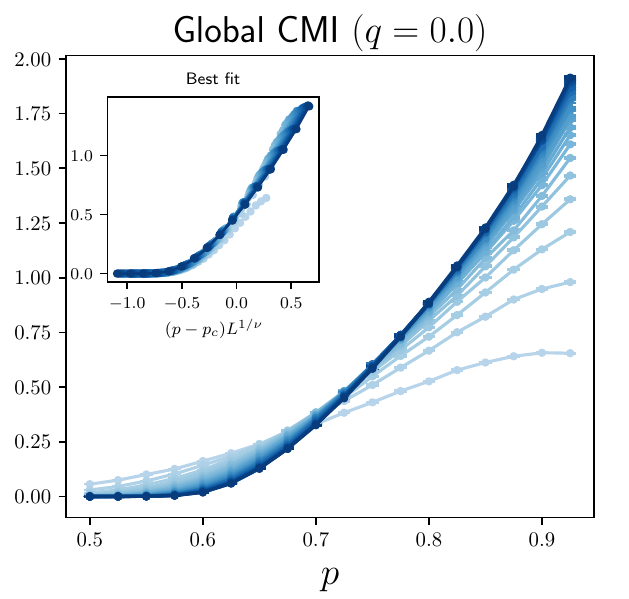}
        \caption{\justifying
        The global CMI as a function of $p$ for fixed $q = 0.0$, and for different sizes of the surface code patch $L\times L$, with $L = 21$--$161$. There is an apparent crossing at $p^\ast \to 0.7$ that drifts only very weakly with $L$. The traces are averaged over $5\times10^4$ samples.}
        \label{fig:apparent_transition_cplc_cmi}
    \end{subfigure}
    \vfill
    \begin{subfigure}[b]{1\linewidth}
        \centering
        \includegraphics[width=0.95\linewidth]{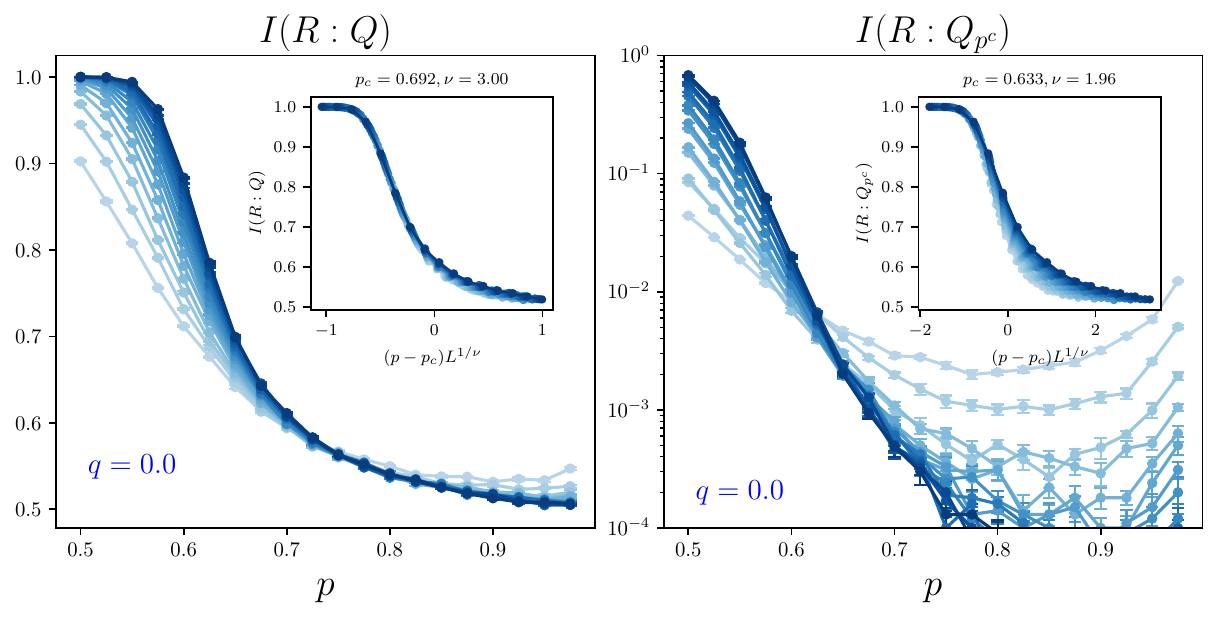}
        \caption{\justifying
        The input–output mutual information of the encoding channel also shows an apparent crossing as a function of $p$ at fixed $q = 0.0$. However, the punctured version shows a crossing at a significantly different $p_c$, suggesting that this transition is a finite-size effect, with strong drift in the scaling curves as $L$ is varied.}
        \label{fig:qec-apparent-transition-cplc}
    \end{subfigure}
    
    \caption{\justifying Apparent critical behaviour at $q = 0.0$. Along the boundary of the CPLC phase diagram, the CMI, and input-output mutual information exhibit an apparent crossing at $p^\ast < 1$, which we interpret as a finite-size crossover controlled by the perfect-metal fixed point at $p=1$ rather than a true phase transition.}
    \label{fig:classical-memory-phase-apparent-transition}
\end{figure}

However, the loop-model description of Nahum \emph{et al.}~\cite{Nahum_2013} implies that along $q=0$ there is, strictly speaking, no transition out of the short loop phase except at the extremal point $p=1$, which we refer to as a perfect ``metal". The states at $q=0$ are always localized for $p<1$, but the associated correlation length $\xi(p)$, which also sets the typical loop size, becomes enormous as one approaches the perfect metal. The sigma-model analysis gives
\begin{equation}
  \xi(p) \sim (1-p)^{-2}\,\exp\!\left[\frac{c}{1-p}\right],
  \label{eq:xi_exponential_perfect_metal}
\end{equation}
for some positive, non-universal constant $c$~\cite{Nahum_2013}. Thus the characteristic length scale grows essentially exponentially fast as $p \to 1$, even though the phase remains localized for any $p<1$.

For a finite system of linear size $L$, the crossover away from the perfect-metal fixed point at $(p=1,q=0)$ occurs when $L$ becomes comparable to $\xi(p)$. This defines an apparent critical point $p^\ast(L)$ along $q=0$ via
\begin{equation}
  L \sim \xi\!\left(p^\ast(L)\right).
\end{equation}
Because of the exponential dependence in Eq.~\eqref{eq:xi_exponential_perfect_metal}, $p^\ast(L)$ drifts towards $p=1$ only extremely slowly as $L$ increases. Over any modest range of system sizes, $p^\ast(L)$ therefore appears almost $L$–independent and mimics a true critical point, exactly as seen in Fig.~\ref{fig:classical-memory-phase-apparent-transition}.

To make this more quantitative, consider two crossings $p^\ast(L_0)$ and $p^\ast(L_1)$ at system sizes $L_0$ and $L_1$, with
\begin{equation}
  L_i \sim \xi\!\left(p^\ast(L_i)\right),\qquad i=0,1.
\end{equation}
Using Eq.~\eqref{eq:xi_exponential_perfect_metal} and eliminating the non-universal prefactor, we obtain:
\begin{equation}
  \ln\frac{L_1}{L_0}
  \simeq
  -2\ln\!\left(\frac{1-p^\ast(L_1)}{1-p^\ast(L_0)}\right)
  + c\left(\frac{1}{1-p^\ast(L_1)} - \frac{1}{1-p^\ast(L_0)}\right).
  \label{eq:Lratio_log}
\end{equation}

As a concrete estimate, suppose that for $L_0 \simeq 150$ we see an apparent crossing at $p^\ast(L_0) \simeq 0.70$, as in Fig.~\ref{fig:classical-memory-phase-apparent-transition}, and we ask how large $L_1$ must be for this apparent crossing to drift to $p^\ast(L_1) \simeq 0.75$. Substituting $p^\ast(L_0)=0.70$ and $p^\ast(L_1)=0.75$ into Eq.~\eqref{eq:Lratio_log} yields
\begin{align}
  \ln\frac{L_1}{L_0}
  &\simeq
  -2\ln\!\left(\frac{1-0.75}{1-0.70}\right)
  + c\left(\frac{1}{1-0.75} - \frac{1}{1-0.70}\right) \nonumber \\
  &\approx 0.36 + 0.67\,c.
\end{align}
The constant $c$ is non-universal and depends on microscopic details; to get an order-of-magnitude estimate we may take $c \sim 1$, in which case
\begin{equation}
  \ln\frac{L_1}{L_0} \approx 1.0
  \qquad\Rightarrow\qquad
  \frac{L_1}{L_0} \sim e^{1.0} \approx 2.8.
\end{equation}
Thus, in order for the apparent crossing to drift from $p^\ast \approx 0.70$ to $p^\ast \approx 0.75$ we would need to increase the system size by roughly a factor of three. Starting from $L_0 \approx 150$, this corresponds to $L_1 \sim 400$. Larger values of $c$ would only increase the required system size. This back-of-the-envelope estimate explains why, over the accessible range $L\lesssim 160$, the apparent critical point along $q=0$ in Fig.~\ref{fig:classical-memory-phase-apparent-transition} appears essentially pinned near $p^\ast \approx 0.7$.

In summary, the numerical results at $q=0$ show an apparent critical behavior---a sharp, almost size-independent crossing in CMI, and input–output mutual information---but the loop-model analysis indicates that this is a finite-size crossover controlled by the perfect-metal fixed point at $p=1$, not a true critical phase. A detailed finite-size scaling of the input–output mutual information (Fig.~\ref{fig:qec-apparent-transition-cplc}) is consistent with a single localized phase along $q=0$ with an exponentially large correlation length near $p=1$.

\end{document}